\documentclass[12pt]{article}
\usepackage{indentfirst}
\usepackage{amssymb}
\usepackage{amsmath}
\usepackage{tabularx}
\usepackage{fontenc}
\usepackage{float}
\usepackage{latexsym}
\usepackage{epsfig}
\usepackage{epstopdf}
\usepackage{indentfirst}
\usepackage{changepage}
\usepackage{graphicx}
\usepackage{subfigure}
\usepackage{breakurl}
\usepackage{hyperref}
\usepackage[numbers,sort&compress]{natbib}

\newtheorem{theorem}{Theorem}

\newtheorem{lemma}[theorem]{Lemma}

\newtheorem{remark}[theorem]{Remark}

\newenvironment{proof}[1][Proof]{\noindent\textbf{#1.} }{\ \rule{0.5em}{0.5em}}
\bibpunct[, ]{[}{]}{,}{n}{,}{,}
\makeatletter
\def\NAT@def@citea{\def\@citea{\NAT@separator}}
\makeatother

\input{tcilatex}

\begin{document}

\title{{\large \textbf{Stochastic Taylor expansion via Poisson point
processes }}}
\author{
\centerline{Weichao Wu\footnote{Weichao Wu, School of Mathematics,
Sun Yat-Sen University, Guangzhou, Guangdong, China, email:
wuwch39@mail.sysu.edu.cn}\hspace{.2cm} and Athanasios C. Micheas\footnote{
Correspondence Author: Athanasios C. Micheas, Department of Statistics,
University of Missouri, 146 Middlebush Hall, Columbia, MO 65211-6100, USA,
email: micheasa@missouri.edu}} }
\maketitle

\begin{abstract}
\noindent We generalize Taylor's theorem by introducing a stochastic
formulation based on an underlying Poisson point process model. We utilize
this approach to propose a novel non-linear regression framework and perform
statistical inference of the model parameters. Theoretical properties of the
proposed estimator are also proven, including its convergence, uniformly
almost surely, to the true function. The theory is presented for the
univariate and multivariate cases, and we exemplify the proposed methodology
using several examples via simulations and an application to stock market
data.
\end{abstract}

\textbf{Keywords: }Function Approximation; Mixture Models; Non-Linear
Regression; Poisson Point Process; Taylor Series; Stochastic Taylor Expansion

\textbf{Mathematics Subject Classifications}: Primary: 41A58, 62J02;
Secondary: 60G55

\section{Introduction}

Function approximation is required in almost every scientific discipline,
from mathematics and statistics, to biology, chemistry, engineering and
physics and in every application of these fields. A commonly used approach
is via Taylor's expansion without remainder, also known as a Taylor
polynomial, which allows us to reproduce a function completely based on
certain derivative values of the target function. In particular, consider
for now the 1-d case and take a real valued function $f:\Re \rightarrow \Re $
that is sufficiently smooth about a point $x_{0}\in \Re $. Then $f$ assumes
a Taylor expansion at the point $x\in \Re $ via
\begin{equation}
f(x)=\sum_{m=0}^{+\infty }\frac{f^{(m)}(x_{0})}{m!}%
(x-x_{0})^{m}=T_{M,x_{0}}(x)+R_{M,x_{0}}(x),  \label{FullTaylorExp}
\end{equation}%
where $T_{M,x_{0}}(x)=\sum_{m=0}^{M}\frac{f^{(m)}(x_{0})}{m!}(x-x_{0})^{m}$
is the $M^{th}$ order Taylor polynomial and $R_{M,x_{0}}(x)$ the remainder
term with the property $R_{M,x_{0}}(x)=O(|x-x_{0}|^{M+1})$ as $%
|x-x_{0}|\rightarrow 0$.

A major challenge that arises when truncating the Taylor series is that
firstly as we move away from a neighborhood of the point $x_0$, the
approximation quality becomes worse and worse, and second, there is no clear
way of choosing the best truncation point $M$. In order to remedy this
problem and still obtain a good approximation at any point $x$, one can turn
to a stochastic formulation of the Taylor expansion so that we treat the
problem as a statistical inference problem. In this way, we can capture the
variability involved in the estimator and approximate the function well at
any point $x$.

Applications of Taylor{'}s theorem include numerical algorithms for
optimization (\cite{More1978}; \cite{Conn2000}), state estimation (\cite%
{Sarkka2013}, Ch. 5), ordinary differential equations (\cite{Hairer1993},
Ch. 2), approximation of exponential integrals in Bayesian statistics (\cite%
{Raudenbush2000}), multivariate kernel approximation (\cite%
{zwicknagl2009power}), and distribution regression models using neural
networks (\cite{shi2025learning}). Algorithms based on Taylor{'}s
approximation that can be viewed as statistical inference problems include
spline interpolation (\cite{Kimeldorf1970,Diaconis1988}), numerical
quadrature (\cite{Diaconis1988,Karvonen2017,Karvonenetal2018}), differential
equations (\cite{Schober2014, Teymur2016, Schober2019}), linear algebra (%
\cite{Hennig2015,Cockayneetal2019}), and function approximation using
Gaussian processes (\cite{Karvonen23}). See the latter paper for additional
discussion and references on Taylor's theorem in mathematics and statistics.

In particular, we can approximate $f(x)$ using the Taylor polynomial $%
T_{M,x_{0}}(x)$ but this requires, unrealistically, knowledge of the
derivatives of $f$ at $x_{0}$. Instead, one can parameterize the problem as
follows
\begin{equation}
\hat{f}(x)=\sum_{m\in \Delta }a_{m}(x-x_{0})^{m},  \label{TaylorPolyEst}
\end{equation}%
where $\Delta =\{0,\dots ,M\}$, and then require estimation of the
parameters $a_{m}$, for $m=1,\dots ,M$, based on observed inputs $x_{i}$ and
corresponding outputs $y_{i}=f(x_{i})$, $i=1,\dots ,K$. More precisely,
consider the non-linear regression model of $y$ on $x$ via
\begin{equation}
y_{i}=\hat{f}(x_{i})+\epsilon _{i},  \label{Regression1d}
\end{equation}%
where the error term can be chosen in different ways, each requiring a
different methodological approach to parameter estimation. In this
formulation, we introduce a first source of randomness that accounts for the
loss of the remainder term. The standard approach is to assume $\epsilon
_{i}\sim N(0,\sigma ^{2})$, $i=1,\dots ,K$, with a common parameter $\sigma
^{2}$ describing the variability in the final estimator.

In this paper, we generalize Taylor's theorem and related existing
approaches for function estimation in the literature, in several ways:

\begin{itemize}
\item We propose methods that treat $\Delta $ as a random set, i.e., we
allow for estimation of the number of terms $M$. Existing methods assume
that the indices set $\Delta $ is deterministic, which includes all standard
and extensions of regression models.

\item The coefficient $a_{m}$ and the power $n_{m}$ of the term $%
a_{m}(x-x_{0})^{n_{m}}$, for $m=1,\dots ,M$, are treated as random real
numbers, and they have their own statistical model.

\item The point process approach we propose allows us to estimate all model
parameters, simultaneously, even though the parameter space can change
dimension; $M$ is random, so that $(a_{1},...,a_{M},n_{1},...,n_{M})$
changes dimension with $M$.

\item The method proposed provides approximation (in the uniformly almost
surely and pointwise sense), for any continuous function, not just
polynomials or analytic functions.

\item We create a novel non-linear regression model, that is highly amenable
to changes in its underlying assumptions.

\item The proposed estimator outperforms most commonly used approaches in
the literature, as the sample size and the dimension of the input $x$
increase, most notably, when performing extrapolation.
\end{itemize}

Specifically, in order to alleviate the aforementioned issues, we propose to
replace the deterministic set $\Delta $ by a random set $\mathcal{N}=$ $%
\{(a_{1},n_{1}),$ $(a_{2},n_{2}),$ $\dots ,$ $(a_{M},n_{M})\}$, so that the
deterministic estimator of equation (\ref{TaylorPolyEst}) is replaced by an
estimator of the Taylor polynomial that contains a second source of
randomness, i.e.,
\begin{equation}
\hat{f}(x)=\sum_{(a,n)\in \mathcal{N}}a(x-x_{0})^{n}.
\label{TaylorPolyPPPEst}
\end{equation}

In this case, the realization of $n$ which corresponds to the power of the
term $(x-x_{0})^{n}$ will be allowed to be a real number in general (not
just an integer), while the coefficient $a$ is viewed as a generalization of
the term $\frac{f^{(n)}(x_{0})}{n!}$ for the given realization of $n$. This
turns $\hat{f}(x)$ into a random variable before introducing the error term
in the regression model, and upon taking expectation with respect to the
random set $\mathcal{N}$, we will have our proposed estimator of $f(x)$ (see
Section \ref{SubSectPropEstimator}).

We will refer to the expansion of equation (\ref{TaylorPolyPPPEst}) as a
Stochastic Taylor Expansion (STE). Clearly, for a specific deterministic
choice of the random set (a constant set) $\mathcal{N}=$ $\{(a=\frac{%
f^{(n)}(x_{0})}{n!},n)\}_{n=0}^{+\infty }$, equation (\ref{TaylorPolyPPPEst}%
) leads to the standard Taylor expansion of equation (\ref{FullTaylorExp}).
Consequently, the proposed STE is a generalization of Taylor's theorem, and
as we will see, it can be used to approximate any continuous function, not
just polynomials or analytic functions. From a mathematical point of view,
equation (\ref{TaylorPolyPPPEst}) allows us to create a much wider class of
functions $f$, which contains all analytic functions, i.e., functions that
can be written as in equation (\ref{FullTaylorExp}).

Combining equations (\ref{Regression1d}) and (\ref{TaylorPolyPPPEst}), we
have the non-linear regression model
\begin{equation}
y_{i}=\sum_{(a,n)\in \mathcal{N}}a(x_{i}-x_{0})^{n}+\epsilon _{i},
\label{TaylorPolyPPPReg}
\end{equation}%
where $\epsilon _{i}\sim N(0,\sigma ^{2})$ (first source of randomness) and $%
\mathcal{N}$ is a random set (second source of randomness), $i=1,\dots ,K$.

The random set $\mathcal{N}$ is a random collection of points $(a,n)\in\Re^2$%
, which is random in number as well, so that standard multivariate analysis
methods cannot be used. Therefore, we turn to point process theory to
provide us with the necessary theoretical framework in order to model,
estimate parameters and study properties of the resulting estimator of the
function $f(x)$.

For foundations, modeling, applications, computation methods and evaluation
of point process models we refer to the texts by \cite{Karr}, \cite{Cressie}%
, \cite{BarndorffNielsen}, \cite{vanLieshout}, \cite%
{lantuejoul2001geostatistical}, \cite{LawsonDenison}, \cite{Moller2003},
\cite{MollerWaag2004}, \cite{DaleyVereJones2005}, \cite{DaleyVereJones2008},
\cite{Illian2008}, \cite{Gelfand2010}, \cite{Chiuetal2013}, \cite%
{Spodarev2013}, \cite{Diggle2013} and \cite{Badd}. Some recent papers
exploring general methodologies, applications and simulations of such
processes include \cite{yau2012generalization}, \cite{cronie2018non}, \cite%
{Micheas2019}, \cite{zhuang2020detection}, \cite{ChenMicheasHolan2020}, \cite%
{TangLi2021}, \cite{Baddeley2022}, \cite{WuMicheas2022}, \cite{Baetal2023},
\cite{KresinSchoenberg2023}, \cite{WuMicheas2024}, \cite{VanLieshout2024},
\cite{cronie2024cross}, \cite{Micheas2025}, and the references therein.

Now, owing to the form of equation (\ref{TaylorPolyPPPReg}), it makes sense
to consider terms $a(x_{i}-x_{0})^{n}$ corresponding to events $(a,n)$
independent of each other and distinct, i.e., we cannot have two
realizations of $n$ yielding the same power and therefore, we take $n\in \Re
$. Consequently, a natural choice for the point process is a model for
independent events, and the most commonly used point process model is the
Poisson point process (e.g., \cite{DaleyVereJones2005}, \cite%
{DaleyVereJones2008}, \cite{Illian2008}, \cite{Diggle2013}, and \cite{Badd}%
). It is often referred to as being \textquotedblleft completely at
random\textquotedblright\ or as a point process with \textquotedblleft no
interactions\textquotedblright , since the number of events (and the events
themselves) over disjoint sets are independent of each other. In this
introductory paper to the STE we consider a flexible choice, that of a
Poisson point process with a mixture model for the intensity function.

The paper proceeds as follows; in Section 2 we discuss the construction of
the estimator $\hat f(x)$ using point process theory and prove that it
converges uniformly almost surely to the true value $f(x)$. In addition, we
prove that the space of functions thus created, is dense in the space of
continuous functions. Section 3 presents the estimation procedure based on
observed data, along with illustrative simulated results for specific
univariate and multivariate cases. An application to stock market data is
presented in Section 4. Concluding remarks are given in the last section.

\section{Taylor Expansion Poisson Point Process Estimator}

\subsection{Poisson Point Processes}

Consider a planar region $\mathcal{W}\subset \Re^{2}$ (extensions to higher
dimensions are straightforward), and suppose that we observe $n$ points
(events) $\varphi _{n}=\{\mathbf{s}_{i}\}_{i=1}^{n}$ from a point process $%
\mathcal{N}$. A realization of a point process is known as a point pattern.
In order to model this collection of points we consider the Inhomogeneous
Poisson point process (IPPP) which requires two assumptions: first, the
random variables $\mathcal{N}(B)$, $B\subseteq \mathcal{W},$ which denote
the number of events over the set $\mathcal{W}$, are distributed as Poisson,
i.e., $\mathcal{N}(B)\thicksim Pois(\Lambda (B)),$ where $\Lambda (B)$ the
intensity measure, and second, counts of events are independent over any
finite collection of disjoint regions.

The intensity measure $\Lambda (B)$ describes the average number of events
over $B$, and is defined by%
\begin{equation}
\Lambda (B)=E\left[ \mathcal{N}(B)\right] =\int\limits_{B}\lambda (\mathbf{s}%
)\mu _{2}(d\mathbf{s}),
\end{equation}%
where $\lambda (\mathbf{s})$ is known as the intensity function or surface
for planar point patterns. The intensity surface $\lambda (\mathbf{s})$
exists via an appeal to the Radon-Nikodym theorem, provided that $\Lambda $
is an absolutely continuous measure with respect to Lebesgue measure $\mu
_{2}$ in $\Re ^{2}$. The special case where $\lambda (\mathbf{s})=\lambda $
yields the homogeneous Poisson process ($HPP$), with intensity $\lambda $
and mean measure $\Lambda (B)=\lambda |B|$, where $|B|=\mu _{2}(B)$ the area
of $B$.

The joint distribution of the points and the number of points $\mathcal{N}(%
\mathcal{W})=n$ over the region $\mathcal{W}$ is given by
\begin{equation}
f(\mathbf{s}_1,\mathbf{s}_2,\dots,\mathbf{s}_n,n)= \frac{e^{-\int_{\mathcal{W%
}}\lambda(\mathbf{s})d\mathbf{s}}} {n!}\prod_{i=1}^n\lambda(\mathbf{s}_i),
n\ge 0.
\end{equation}

Clearly, in order to uniquely determine the IPPP we require a model for the
intensity $\lambda(\mathbf{s})$. Several models can be considered, however,
in order to create a flexible model that can capture a plethora of cases for
the Taylor polynomial, we consider a parametric model based on an $M$%
-component mixture model, with bivariate normal components. This
construction for the intensity function was illustrated in \cite%
{ChakrabortyGelfand}, \cite{Micheas2014} and \cite{Micheas2019}.

In particular, the model for the intensity function is as follows:
\begin{equation}
\lambda(\mathbf{s}|\boldsymbol{\theta})=\lambda\sum_{m=1}^M p_m g_m(\mathbf{s%
}|\boldsymbol{\mu}_m,\boldsymbol{\Sigma}_m),  \label{IntSurfFull}
\end{equation}
where $p_m$ is the probability that the point arises from the $m^{th}$
mixture component, $p_m\ge 0$ and $\sum_{m=1}^Mp_m=1$, with $g_m(\mathbf{s}|%
\boldsymbol{\mu}_m,\boldsymbol{\Sigma}_m)$ denoting the density function of
the $m^{th}$ bivariate normal component, with mean of $\boldsymbol{\mu}_m$
and variance of $\boldsymbol{\Sigma}_m$. Finally, $\lambda>0$, is a constant
describing the average number of events over the region $\mathcal{W}$,
provided that $g_m(.)$ is a proper density (wlog almost surely) over $%
\mathcal{W}$, since then we can write
\begin{equation}
E[\mathcal{N}(\mathcal{W})]=\Lambda(\mathcal{W})= \int_\mathcal{W}\lambda(%
\mathbf{s}|\boldsymbol{\theta}) d\mathbf{s}=\lambda\sum_{m=1}^M p_m \int_%
\mathcal{W}g_m(\mathbf{s}|\boldsymbol{\mu}_m,\boldsymbol{\Sigma}_m) d\mathbf{%
s}=\lambda.  \label{e3.1}
\end{equation}
Next, we discuss the construction of the Taylor expansion Poisson point
process estimator.

\subsection{Proposed Estimator: Univariate Case\label{SubSectPropEstimator}}

Consider an observation window $\mathcal{W}\subset \Re ^{2}$ and a
realization of the IPPP with intensity surface given by equation (\ref%
{IntSurfFull}) over $\mathcal{W}$, say, $\varphi
_{v}=\{(a_{1},n_{1}),(a_{2},n_{2}),...,(a_{v},n_{v})\}$, with $\mathcal{N}(%
\mathcal{W})=v$ and $\mathbf{s}=(a,n)$. In the intensity function of
equation (\ref{IntSurfFull}), let $\mathbf{\mu }_{m}=%
\begin{pmatrix}
\mu _{a,m} \\
\mu _{n,m}%
\end{pmatrix}%
\in \Re ^{2}$, and $\mathbf{\Sigma} _{m}=%
\begin{pmatrix}
\sigma _{a,m}^{2} & \rho _{m}\sigma _{n,m}\sigma _{a,m} \\
\rho _{m}\sigma _{n,m}\sigma _{a,m} & \sigma _{n,m}^{2}%
\end{pmatrix}%
$, where $|\rho _{m}|<1$, and $\sigma _{n,m},\sigma _{a,m}>0$, for all $%
m=1,2,\dots ,M$.

Without loss of generality, and in order to simplify the exposition that
follows, we will consider $x>x_0$. The following theorem provides the form
of the estimator in terms of the parameters of the underlying Poisson point
process.

\begin{theorem}[Taylor Expansion via IPPP]
\label{TaylorExpansionMeanThm} Consider the random variables
\begin{equation}
\hat{f}_{\mathcal{N}}(x)=\sum_{(a,n)\in \mathcal{N}}a(x-x_{0})^{n},
\end{equation}%
where $\mathcal{N}$ denotes an IPPP with the intensity surface of equation (%
\ref{IntSurfFull}), and assume that $x>x_{0}$. Then the Taylor expansion
Poisson point process estimator (TPE) of the function $f(x)$ is given by
\begin{equation}
\hat{f}_{TPE}^{M}(x)=E(\hat{f}_{\mathcal{N}}(x))=\lambda
\sum_{m=1}^{M}p_{m}(\mu _{a,m}+\rho _{m}\sigma _{a,m}\sigma _{n,m}\ln
(x-x_{0}))(x-x_{0})^{\mu _{n,m}+\frac{\ln (x-x_{0})\sigma _{n,m}^{2}}{2}}
\label{TPEgen1}
\end{equation}
\end{theorem}

\begin{proof}
Let $h(n,a)=a(x-x_{0})^{n}$, and write
\begin{equation*}
\hat{f}_{\mathcal{N}}(x)=\sum_{(a,n)\in \mathcal{N}}a(x-x_{0})^{n}=%
\sum_{(a,n)\in \mathcal{N}}h(n,a).
\end{equation*}

By Campbell's theorem (\cite{r47}) for point process sums, we have
\begin{equation*}
\begin{split}
E(\hat{f}_{\mathcal{N}}(x))& =E\left( \sum_{(a,n)\in \mathcal{N}%
}h(n,a)\right) =\int_{\Re ^{2}}h(a,n)\lambda (a,n)dadn \\
& =\int_{\Re ^{2}}a(x-x_{0})^{n}\lambda \sum_{m=1}^{M}p_{m}g_{m}(a,n|\mathbf{%
\mu }_{m},\mathbf{\Sigma }_{m})dadn
\end{split}%
\end{equation*}%
so that%
\begin{equation}
E(\hat{f}_{\mathcal{N}}(x))=\lambda \sum_{m=1}^{M}p_{m}\int\limits_{\Re
^{2}}a(x-x_{0})^{n}g_{m}(a,n|\mathbf{\mu }_{m},\mathbf{\Sigma }_{m})dadn,
\label{theoremmeanTPEeq1}
\end{equation}%
and we require calculation of the double integral%
\begin{equation*}
I=\int\limits_{\Re ^{2}}a(x-x_{0})^{n}g_{m}(a,n|\mathbf{\mu }_{m},\mathbf{%
\Sigma }_{m})dadn,
\end{equation*}%
where $a,n|\mathbf{\mu }_{m},\mathbf{\Sigma }_{m}$ is a bivariate normal $%
N_{2}(\mathbf{\mu }_{m},\mathbf{\Sigma }_{m})$, with $\mathbf{\mu }_{m}=(\mu
_{a,m},\mu _{n,m})^{T}$, and $\mathbf{\Sigma }_{m}=$ \newline
$%
\begin{pmatrix}
\sigma _{a,m}^{2} & \rho _{m}\sigma _{n,m}\sigma _{a,m} \\
\rho _{m}\sigma _{n,m}\sigma _{a,m} & \sigma _{n,m}^{2}%
\end{pmatrix}%
$. Since%
\begin{equation*}
g_{m}(a,n|\mathbf{\mu }_{m},\mathbf{\Sigma }_{m})=g_{m}(n|\mu _{n,m},\sigma
_{n,m}^{2})g_{m}(a|n,\mathbf{\mu }_{m},\mathbf{\Sigma }_{m}),
\end{equation*}%
with $n|\mu _{n,m},\sigma _{n,m}^{2}\backsim N(\mu _{n,m},\sigma _{n,m}^{2})$
and $a|n,\mathbf{\mu }_{m},\mathbf{\Sigma }_{m}\backsim N(\mu _{a,m}+$ $\rho
_{m}(n-\mu _{n,m})\frac{\sigma _{a,m}}{\sigma _{n,m}},$ $(1-\rho _{m})^{2}$ $%
\sigma _{a,m}^{2}),$ we have%
\begin{eqnarray*}
I &=&\int\limits_{\Re }(x-x_{0})^{n}g_{m}(n|\mu _{n,m},\sigma _{n,m}^{2})
\left[ \int\limits_{\Re }ag_{m}(a|n,\mathbf{\mu }_{m},\mathbf{\Sigma }_{m})da%
\right] dn \\
&=&\int\limits_{\Re }(x-x_{0})^{n}g_{m}(n|\mu _{n,m},\sigma _{n,m}^{2})E
\left[ a|n,\mathbf{\mu }_{m},\mathbf{\Sigma }_{m})\right] dn \\
&=&\int\limits_{\Re }(x-x_{0})^{n}\left[ \mu _{a,m}+\rho _{m}(n-\mu _{n,m})%
\frac{\sigma _{a,m}}{\sigma _{n,m}}\right] g_{m}(n|\mu _{n,m},\sigma
_{n,m}^{2})dn,
\end{eqnarray*}%
and therefore%
\begin{equation}
I=\mu _{a,m}E\left[ (x-x_{0})^{n}\right] +\rho _{m}(E\left[ \left( n-\mu
_{n,m}\right) (x-x_{0})^{n}\right] )\frac{\sigma _{a,m}}{\sigma _{n,m}}.
\label{UnivI}
\end{equation}%
Since $x-x_{0}>0$, we have $(x-x_{0})^{n}=\exp (n\ln (x-x_{0}))$, and thus
\begin{eqnarray*}
&&(x-x_{0})^{n}\exp \left\{ -\frac{1}{2\sigma _{n,m}^{2}}\left( n-\mu
_{n,m}\right) ^{2}\right\} \\
&=&\exp \left\{ -\frac{1}{2\sigma _{n,m}^{2}}\left( n^{2}-2n\mu _{n,m}+\mu
_{n,m}^{2}-2\sigma _{n,m}^{2}n\log (x-x_{0})\right) \right\} \\
&=&\exp \left\{ -\frac{1}{2\sigma _{n,m}^{2}}\left( n^{2}-2[\mu
_{n,m}+\sigma _{n,m}^{2}\log (x-x_{0})]n+\mu _{n,m}^{2}\right) \right\} \\
&=&\exp \left\{ -\frac{1}{2\sigma _{n,m}^{2}}(n-[\mu _{n,m}+\sigma
_{n,m}^{2}\log (x-x_{0})])^{2}\right\} \\
&&\exp \left\{ -\frac{1}{2\sigma _{n,m}^{2}}(\mu _{n,m}^{2}-[\mu
_{n,m}+\sigma _{n,m}^{2}\log (x-x_{0})]^{2})\right\} ,
\end{eqnarray*}%
so that%
\begin{eqnarray*}
&&E\left[ (x-x_{0})^{n}|n\backsim N(\mu _{n,m},\sigma _{n,m}^{2})\right]
=\exp \left\{ -\frac{1}{2\sigma _{n,m}^{2}}(\mu _{n,m}^{2}-[\mu
_{n,m}+\sigma _{n,m}^{2}\log (x-x_{0})]^{2})\right\} \\
&&\int\limits_{\Re }\frac{1}{\sqrt{2\pi \sigma _{n,m}^{2}}}\exp \left\{ -%
\frac{1}{2\sigma _{n,m}^{2}}(n-[\mu _{n,m}+\sigma _{n,m}^{2}\log
(x-x_{0})])^{2}\right\} dn \\
&=&\exp \left\{ -\frac{1}{2\sigma _{n,m}^{2}}(\mu _{n,m}^{2}-\mu
_{n,m}^{2}-2\mu _{n,m}\sigma _{n,m}^{2}\log (x-x_{0})-\sigma
_{n,m}^{4}\left( \log (x-x_{0})\right) ^{2})\right\} ,
\end{eqnarray*}%
which leads to%
\begin{eqnarray}
E\left[ (x-x_{0})^{n}|\mu _{n,m},\sigma _{n,m}^{2}\right] &=&\exp \left\{
\log (x-x_{0})(\mu _{n,m}+\sigma _{n,m}^{2}\log (x-x_{0}))\right\}
\label{Exp1} \\
&=&(x-x_{0})^{\mu _{n,m}+\frac{1}{2}\sigma _{n,m}^{2}\log (x-x_{0})}.  \notag
\end{eqnarray}%
Similarly, we write%
\begin{eqnarray*}
&&E\left[ \left( n-\mu _{n,m}\right) (x-x_{0})^{n}|n\backsim N(\mu
_{n,m},\sigma _{n,m}^{2})\right] \\
&=&\exp \left\{ -\frac{1}{2\sigma _{n,m}^{2}}(\mu _{n,m}^{2}-[\mu
_{n,m}+\sigma _{n,m}^{2}\log (x-x_{0})]^{2})\right\} \\
&&\int\limits_{\Re }\left( n-\mu _{n,m}\right) \frac{1}{\sqrt{2\pi \sigma
_{n,m}^{2}}}\exp \left\{ -\frac{1}{2\sigma _{n,m}^{2}}(n-[\mu _{n,m}+\sigma
_{n,m}^{2}\log (x-x_{0})])^{2}\right\} dn \\
&=&(x-x_{0})^{\mu _{n,m}+\frac{1}{2}\sigma _{n,m}^{2}\log (x-x_{0})}E(n-\mu
_{n,m}|n\backsim N(\mu _{n,m}+\sigma _{n,m}^{2}\log (x-x_{0}),\sigma
_{n,m}^{2}),
\end{eqnarray*}%
and therefore%
\begin{eqnarray}
&&E\left[ \left( n-\mu _{n,m}\right) (x-x_{0})^{n}|n\backsim N(\mu
_{n,m},\sigma _{n,m}^{2})\right]  \label{Exp2} \\
&=&\left[ \sigma _{n,m}^{2}\log (x-x_{0})\right] (x-x_{0})^{\mu _{n,m}+\frac{%
1}{2}\sigma _{n,m}^{2}\log (x-x_{0})}.  \notag
\end{eqnarray}%
Using (\ref{Exp1}) and (\ref{Exp2}) in (\ref{UnivI}), we have
\begin{eqnarray*}
I &=&\mu _{a,m}(x-x_{0})^{\mu _{n,m}+\frac{1}{2}\sigma _{n,m}^{2}\log
(x-x_{0})}+ \\
&&\rho _{m}\left( \sigma _{n,m}^{2}\log (x-x_{0})(x-x_{0})^{\mu _{n,m}+\frac{%
1}{2}\sigma _{n,m}^{2}\log (x-x_{0})}\right) \frac{\sigma _{a,m}}{\sigma
_{n,m}} \\
&=&\mu _{a,m}(x-x_{0})^{\mu _{n,m}+\sigma _{n,m}^{2}\log (x-x_{0})} \\
&&+\rho _{m}\sigma _{a,m}\sigma _{n,m}\log (x-x_{0})(x-x_{0})^{\mu _{n,m}+%
\frac{1}{2}\sigma _{n,m}^{2}\log (x-x_{0})},
\end{eqnarray*}%
and therefore%
\begin{equation*}
I=\left[ \mu _{a,m}+\rho _{m}\sigma _{a,m}\sigma _{n,m}\log (x-x_{0})\right]
(x-x_{0})^{\mu _{n,m}+\frac{1}{2}\sigma _{n,m}^{2}\log (x-x_{0})},
\end{equation*}%
which leads to the desired result of equation (\ref{TPEgen1}).
\end{proof}

Since we do not want the individual components driving the analysis, we
consider the equally likely case $p_1=p_2=\dots=p_M=\frac{1}{M}$ and further
simplify the intensity by letting $\lambda=M$, so that the intensity
function of equation (\ref{IntSurfFull}) reduces to
\begin{equation}
\lambda(\mathbf{s}|\boldsymbol{\theta}_M)= \lambda(a,n|\boldsymbol{\theta}%
)=\sum_{m=1}^M g_m(\mathbf{s}|\boldsymbol{\mu}_m,\boldsymbol{\Sigma}_m),
\label{IntSurfReduced}
\end{equation}
where $\boldsymbol{\theta}_M=$ $(\mu_{a,1},\mu_{a,2},$ $\dots,\mu_{a,M},$ $%
\mu_{n,1},\mu_{n,2},$ $\dots,\mu_{n,M},$ $\rho_1,\rho_2,\dots$ $%
,\rho_M,\sigma_{a,1},$ $\sigma_{a,2},\dots,$ $\sigma_{a,M},\sigma_{n,1},$ $%
\sigma_{n,2},\dots,\sigma_{n,M})$, denote the parameters of the intensity
surface. This construction leads to having $M$ events, on the average, in
realizations of the IPPP, and allows us to control the number of terms in
the expansion. Clearly, the larger the value of $M$, the better the
approximation.

The TPE under the intensity function (\ref{IntSurfReduced}) reduces to
\begin{equation}
\hat{f}^{M,\boldsymbol{\theta}_M}_{TPE}(x)=\sum_{m=1}^M( \mu_{a,m}+
\rho_m\sigma_{a,m}\sigma_{n,m}\ln(x-x_0))(x-x_0)^ {\mu_{n,m}+\frac{%
\ln(x-x_0)\sigma_{n,m}^2}{2}}.  \label{TPEmean}
\end{equation}

In view of the latter equation, a few remarks are in order.

\begin{remark}
We note the following.

\begin{itemize}
\item {The TPE as constructed offers two intriguing directions that require
further investigation. Firstly, we notice that for different choices of the
parameters of the intensity surface, we can construct a different function
via the TPE $\hat{f}^{M,\boldsymbol{\theta}_M}_{TPE}(x)$. This is a
mathematical point of view of the construction, where the TPE allows us to
create a novel class of functions indexed by $\boldsymbol{\theta}_M$. On the
other hand, we have the statistical point of view, in which we would like to
estimate the values of the parameters $\boldsymbol{\theta}_M$, based on
observed inputs and outputs of the function $f$.}

\item {At first glance it may appear that the TPE consists of a finite
number of terms. However, since the $\log $ function appears in the
coefficient $a$ and power $n$ of the term $a(x-x_{0})^{n}$, the TPE involves
an infinite number of terms.}

\item {Intuitively, consider equation (\ref{TPEmean}) with $M+1$ terms, and
set $\rho_m=0$, $\sigma_{n,m}=0$, $\mu_{a,m}=\frac{f^{(m)}(x_0)}{m!}$, and $%
\mu_{n,m}=m$, $m=1,\dots,M$, with $\rho_{M+1}=0$, $\sigma_{n,{M+1}}=0$, $%
\mu_{a,{M+1}}=1$, and $\mu_{n,{M+1}}=0$. Then sending $M\rightarrow\infty$,
we obtain the standard Taylor expansion as a special case. This illustrates
that there is always a set of parameter values that will take us to the true
function. Of course, in practice, we will never get this perfect situation,
and this is seen in our simulation section, where we might have a larger
number of estimated number of terms $\hat M$ than the true $M$, but the
estimated parameters adjust to give us near perfect fits.}

\item {In contrast to existing methods for function estimation, we notice
that the TPE does not require us, for example, to select a bandwidth (kernel
based methods) or the number of knots (spline methods), which makes our
approach easier to use.}
\end{itemize}
\end{remark}

First we consider the mathematical implications of the proposed TPE. In
particular, we provide some insight on the space of functions created via
equation (\ref{TPEmean}) in the following lemma.

\begin{lemma}[Dense TPE Space]
\label{DenseFunctionLemma}Let $\Delta _{L}=[x_{0},x_{0}+L]$ a compact
interval, $L>0$, $x_{0}\in \Re $, and let ${\mathcal{C}}_{\Delta _{L}}^{\Re
} $ denote the space of all continuous functions from $\Delta _{L}$ into $%
\Re $. Let the space of functions of equation (\ref{TPEmean}), $M=1,2,\dots $%
, indexed by $\boldsymbol{\theta}_{M},$ be denoted by ${\mathcal{F}}_{\Delta
_{L}}^{\Re }$ and denote by ${\mathcal{E}}_{\Delta _{L}}^{\Re }\subseteq {%
\mathcal{F}}_{\Delta _{L}}^{\Re }$ the subset of functions with $\rho _{m}=0$%
, $m=0,1,\dots ,M$. Then ${\mathcal{E}}_{\Delta _{L}}^{\Re }$ is dense in ${%
\mathcal{C}}_{\Delta _{L}}^{\Re }$, i.e., any real, continuous function over
$\Delta _{L}$ can be represented as the limit of a sequence with members
from ${\mathcal{E}}_{\Delta _{L}}^{\Re }$, and consequently, as the limit of
members from ${\mathcal{F}}_{\Delta _{L}}^{\Re }$.
\end{lemma}

\begin{proof}
\label{DenseFunctionLemmaProof}Suppose that $f\in {\mathcal{C}}_{\Delta
_{L}}^{\Re }$, where $f$ is the function of interest and let $\hat{f}%
_{TPE}^{M,\boldsymbol{\theta}_{M}}\in {\mathcal{F}}_{\Delta _{L}}^{\Re }$,
that is,
\begin{equation}
\hat{f}_{TPE}^{M,\boldsymbol{\theta}_{M}}(x)=\sum_{m=1}^{M}(\mu _{a,m}+\rho
_{m}\sigma _{a,m}\sigma _{n,m}\ln (x-x_{0}))(x-x_{0})^{\mu _{n,m}+\frac{\ln
(x-x_{0})\sigma _{n,m}^{2}}{2}},  \label{fTPE}
\end{equation}%
where $M\in \mathbb{N}^{+}$, $\boldsymbol{\theta_M}=(\mu _{a,1},\mu
_{a,2},...,\mu _{a,M},\mu _{n,1},\mu _{n,2},...,\mu _{n,M},$ $\rho _{1},\rho
_{2},...,\rho _{M},$ $\sigma _{a,1},\sigma _{a,2},...,\sigma _{a,M},$ $%
\sigma _{n,1},\sigma _{n,2},...,\sigma _{n,M})$, $\mu _{a,m}\in \mathbb{R}%
,\mu _{n,m}\in \mathbb{R},\rho _{m}\in \lbrack -1,1],$ $\sigma _{a,m}\geq 0,$
$\sigma _{n,m}\geq 0,$ for $m=1,2,...,M$. Moreover, we denote by ${\mathcal{E%
}}_{\Delta _{L}}^{\Re }$ the space of functions such that $\rho _{m}=0$ for $%
m=1,2,...,M$, i.e.,
\begin{equation}
\hat{f}_{TPE}^{M,\boldsymbol{\theta}_{M}}(x)=\sum_{m=1}^{M}\mu
_{a,m}(x-x_{0})^{\mu _{n,m}+\frac{\ln (x-x_{0})\sigma _{n,m}^{2}}{2}},
\end{equation}%
so that ${\mathcal{E}}_{\Delta _{L}}^{\Re }\subseteq {\mathcal{F}}_{\Delta
_{L}}^{\Re }$. Although ${\mathcal{F}}_{\Delta _{L}}^{\Re }$ is not a
sub-algebra of ${\mathcal{C}}_{\Delta _{L}}^{\Re }$, we can easily see that $%
{\mathcal{E}}_{\Delta _{L}}^{\Re }$ is, since first,
\begin{eqnarray*}
&&\mu _{a,1}(x-x_{0})^{\mu _{n,1}+\frac{\sigma _{n,1}^{2}}{2}\ln
(x-x_{0})}\mu _{a,2}(x-x_{0})^{\mu _{n,2}+\frac{\sigma _{n,2}^{2}}{2}\ln
(x-x_{0})} \\
&=&(\mu _{a,1}\mu _{a,2})(x-x_{0})^{(\mu _{n,1}+\mu _{n,2})+\frac{\sigma
_{n,1}^{2}+\sigma _{n,2}^{2}}{2}\ln (x-x_{0})}\in {\mathcal{E}}_{\Delta
_{L}}^{\Re },
\end{eqnarray*}%
second, there exists a function $f(x)=1\in {\mathcal{E}}_{\Delta _{L}}^{\Re
} $ (take let $M=1,$ $\mu _{a,1}=1,$ $\mu _{n,1}=\sigma _{n,1}=0$), and
third, ${\mathcal{E}}_{\Delta _{L}}^{\Re }$ contains functions that separate
points, i.e., for any $x,y\in \Re $, there is a function $f\in {\mathcal{E}}%
_{\Delta _{L}}^{\Re }$, such that $f(x)\neq f(y)$. Thus, by the
Stone-Weierstrass Theorem, ${\mathcal{E}}_{\Delta _{L}}^{\Re }$ is dense in $%
{\mathcal{C}}_{\Delta _{L}}^{\Re }$, and therefore ${\mathcal{F}}_{\Delta
_{L}}^{\Re }$ is dense on ${\mathcal{C}}_{\Delta _{L}}^{\Re }$.\newline
As a result, for any function $f\in {\mathcal{C}}_{\Delta _{L}}^{\Re }$, and
a given $\epsilon >0$, there exists a sequence of functions $f_{TPE}^{M,\hat{%
\boldsymbol{\theta}}_{M}}\in {\mathcal{F}}_{\Delta _{L}}^{\Re },$ such that
\begin{equation}
|\hat{f}^{M,\boldsymbol{\theta}_{M}}(x)-f(x)|<\epsilon ,
\end{equation}%
for all $x\in \Delta _{L}$, and
\begin{equation}
\lim_{M\rightarrow \infty }\hat{f}_{TPE}^{M,\boldsymbol{\hat{\theta}}%
_{M}}(x)=f(x)  \label{LimitDenseTPE}
\end{equation}
\end{proof}

The latter suggests that any continuous function can be approximated using
the proposed framework, not just polynomials or analytic functions, which is
a generalization of Taylor's theorem.

Next we turn to treating the other implied direction of the TPE, the
statistical inference framework; when inputs and outputs of a function are
given, we wish to estimate the coefficients of the TPE, i.e., the parameters
of the underlying Poisson point process model. This will provide the
function at all $x$ and allow for prediction while accounting for the
variability involved. This approach helps us create a novel non-linear
regression framework, which is one of the major contributions of this work
and the proposed TPE.

Using the observed data $(x_{k},y_{k}=f(x_{k}))$, $k=1,\dots,K$, and the TPE
of equation (\ref{TPEmean}), we have the non-linear regression model
\begin{equation}
y_k= \sum_{m=1}^M(\mu_{a,m}+
\rho_m\sigma_{a,m}\sigma_{n,m}\ln(x_k-x_0))(x_k-x_0)^ {\mu_{n,m}+\frac{%
\ln(x_k-x_0)\sigma_{n,m}^2}{2}}+\epsilon_k,  \label{TaylorNonLinearReg}
\end{equation}
where $\epsilon _{k}\sim N(0,\sigma^{2})$, and $x_k>x_0$, $k=1,\dots,K$.
Naturally, we set $x_0=min(x_1,\dots,x_K)-\delta$, for some small $\delta>0$.

Next we discuss the asymptotic properties of the estimator $\hat{f}^{M,%
\boldsymbol{\hat\theta}_M}_{TPE}(x)$.

\subsection{Proposed Estimator: Convergence}

Consider the maximum likelihood estimators (MLEs) $\hat{\boldsymbol{\theta}}%
_M$ of the parameters $\boldsymbol{\theta}_M$ in the TPE $\hat{f}^{M,%
\boldsymbol{\theta}_M}_{TPE}(x)$ and $\hat{\sigma}^{2}$, the MLE of $%
\sigma^{2}$. For a given integer $M$, the parameters of $\boldsymbol{\theta}%
_M$ requiring estimation are $\mu_{a,m}$, $\rho_m$, $\sigma_{a,m}$, $%
\sigma_{n,m}$, and $\mu_{n,m}$, for $m=1,2,\dots,M$, and they are chosen in
order to maximize the likelihood function of the non-linear regression model
of equation (\ref{TaylorNonLinearReg}).

First we show that $\hat{f}_{TPE}^{M,\hat{\boldsymbol{\theta}}_{M}}(x)$ is a
strongly consistent estimator of $\hat{f}_{TPE}^{M,\boldsymbol{\theta}%
_{M}}(x)$, pointwise, as $K\rightarrow +\infty $.

\begin{theorem}[Pointwise Almost Sure Convergence]
\label{LSEConvergence} Assume that $x>x_0$, and let $\hat{\boldsymbol{\theta}%
}_M$ and $\hat{\sigma}^{2}$ denote the MLEs of $\boldsymbol{\theta}_M$ and $%
\sigma^{2}$, respectively, based on the non-linear regression model of
equation (\ref{TaylorNonLinearReg}), and under the ordering
\begin{equation}
\mu_{a,1}<\mu_{a,2}<\dots<\mu_{a,M}.  \label{Condition1}
\end{equation}
Then $\hat{f}^{M,\hat{\boldsymbol{\theta}}_M}_{TPE}(x)$ converges almost
surely to the proposed estimator $\hat{f}^{M,\boldsymbol{\theta}_M}_{TPE}(x)$
of equation (\ref{TPEmean}) pointwise in $x$, i.e.,
\begin{equation}
\hat f^{M,\hat{\boldsymbol{\theta}}_M}_{TPE}(x) \overset{ a.s.}{\rightarrow}
\hat f^{M,\boldsymbol{\theta}_M}_{TPE}(x),  \label{TPEasConv}
\end{equation}
for all $x>x_0$, as $K\rightarrow +\infty$, and
\begin{equation}
\hat{\sigma}^{2}\overset{ a.s.}{\rightarrow}\sigma^{2}.
\end{equation}
\end{theorem}

\begin{proof}
The likelihood function based on the non-linear regression model of equation
(\ref{TaylorNonLinearReg}) is given by
\begin{equation}
\emph{L}(\boldsymbol{\theta}_{M},\sigma ^{2}|\boldsymbol{x},\boldsymbol{y}%
)=(2\pi \sigma ^{2})^{-K/2}\exp \left\{ -\frac{1}{2\sigma ^{2}}%
\sum\limits_{k=1}^{K}(y_{k}-\hat{f}_{TPE}^{M,\boldsymbol{\theta}%
_{M}}(x_{k}))^{2}\right\} ,
\end{equation}%
where $\boldsymbol{\theta}_{M}=(\mu _{a,1},\mu _{a,2},\dots ,\mu _{a,M},\mu
_{n,1},\mu _{n,2},\dots ,\mu _{n,M},\rho _{1},\rho _{2},\dots ,\rho
_{M},\sigma _{a,1},\sigma _{a,2},\dots ,$ \newline
$\sigma _{a,M},\sigma _{n,1},\sigma _{n,2},\dots ,\sigma _{n,M})$, $%
\boldsymbol{x}=(x_{1},\dots ,x_{K})$ and $\boldsymbol{y}=(y_{1},\dots
,y_{K}) $.\newline
Clearly, the MLE of $\sigma ^{2}$ is given in closed form by
\begin{equation}
\hat{\sigma}^{2}=\frac{1}{K}\sum\limits_{k=1}^{K}(y_{k}-\hat{f}_{TPE}^{M,%
\boldsymbol{\hat\theta}_{M}}(x_{k}))^{2},
\end{equation}%
where $\boldsymbol{\hat\theta}_{M}$, the MLE of $\boldsymbol{\theta}_{M}$
can only be obtained numerically.\newline
A straightforward application of Theorem 17 in \cite{Ferguson}, gives
\begin{equation}
\hat{\boldsymbol{\theta}}_{M}\overset{a.s.}{\rightarrow }\boldsymbol{{%
\theta}}_{M},
\end{equation}%
and
\begin{equation}
\hat{\sigma}^{2}\overset{a.s.}{\rightarrow }\sigma ^{2}.
\end{equation}%
Note here that the ordering condition (\ref{Condition1}) guarantees that any
non-identifiability issues with the estimator are alleviated, so that the
identifiability requirement of Theorem 17, requirement 5 is satisfied.%
\newline
Since $\hat{f}_{TPE}^{M,\boldsymbol {\theta}_{M}}(x)$ in equation (\ref%
{TPEmean}) is continuous in $\boldsymbol{\theta}_{M}$ for each fixed $x\in
\Re $, we can write
\begin{equation}
\hat{f}_{TPE}^{M,\hat{\boldsymbol{\theta}}_{M}}(x)\overset{a.s.}{\rightarrow
}\hat{f}_{TPE}^{M,\boldsymbol{\theta}_{M}}(x),
\end{equation}%
for all $x\in \Re $, as $K\rightarrow +\infty $.
\end{proof}

The following theorem presents the uniform almost sure convergence of the
estimated TPE to the true function, as $K,M\rightarrow +\infty $.

\begin{theorem}[$\mathcal{L}^{1}$ and Uniformly Almost Sure Convergence]
\label{TPEConvergence} Assume that $x>x_{0}$, and let $\hat{%
\boldsymbol{\theta}}_{M}$ and $\hat{\sigma}^{2}$ as in Theorem \ref%
{LSEConvergence}, and let the integrated distance ($\mathcal{L}^{1}-norm$)
between the estimator $\hat{f}_{TPE}^{M,\hat{\boldsymbol{\theta}}_{M}}(x)$
and the true value $f(x)$ be defined by
\begin{equation}
D_{M}=\int_{x_{0}}^{+\infty }|\hat{f}_{TPE}^{M,\hat{\boldsymbol{\theta}}%
_{M}}(x)-f(x)|dx.
\end{equation}%
Then $D_{M}\rightarrow 0$, i.e.,
\begin{equation}
\hat{f}_{TPE}^{M,\hat{\boldsymbol{\theta}}_{M}}(x)\overset{\mathcal{L}^{1}}{%
\rightarrow }f(x),
\end{equation}%
as $K,M\rightarrow +\infty $ and
\begin{equation}
P(\sup_{x>x_{0}}|\hat{f}_{TPE}^{M,\hat{\boldsymbol{\theta}}%
_{M}}(x)-f(x)|\rightarrow 0)=1,
\end{equation}%
i.e., $\hat{f}_{TPE}^{M,\hat{\boldsymbol{\theta}}_{M}}(x)$ converges
uniformly almost surely to the function $f(x)$, for all $x>x_{0}$, as $%
K,M\rightarrow +\infty $.
\end{theorem}

\begin{proof}
Consider the $\mathcal{L}^{1}-norm$ between the estimator $\hat{f}_{TPE}^{M,%
\hat{\boldsymbol{\theta}}_{M}}(x)$ and the true value continuous function $%
f(x)$ over the interval of $\Delta _{L}=[x_{0},x_{0}+L]$, $L>0$, defined by
\begin{equation}
D_{M}^{L}=\int_{x_{0}}^{x_{0}+L}|\hat{f}_{TPE}^{M,\hat{\boldsymbol{\theta}}%
_{M}}(x)-f(x)|dx,
\end{equation}%
and note that $D_{M}^{L}\rightarrow D_{M},$ as $L\rightarrow \infty ,$ and
in addition%
\begin{equation}
0\leq \int_{x_{0}}^{x_{0}+L}|\hat{f}_{TPE}^{M,\hat{\boldsymbol{\theta}}%
_{M}}(x)-f(x)|dx\leq L\sup_{x\in \Delta _{L}}|\hat{f}_{TPE}^{M,\hat{%
\boldsymbol{\theta}}_{M}}(x)-f(x)|<+\infty .  \label{Dmbound}
\end{equation}%
Define
\begin{equation}
x_{Max}=\arg \sup_{x\in \Delta _{L}}|\hat{f}_{TPE}^{M,\hat{%
\boldsymbol{\theta}}_{M}}(x)-f(x)|,  \label{XmaxDef}
\end{equation}%
and as a result of equation (\ref{TPEasConv}), we have
\begin{equation}
\hat{f}_{TPE}^{M,\hat{\boldsymbol{\theta}}_{M}}(x_{Max})\overset{p}{%
\rightarrow }\hat{f}_{TPE}^{M,\boldsymbol{\theta}_{M}}(x_{Max}).
\end{equation}%
Therefore, for any $\delta >0$ we can write
\begin{equation}
\lim_{K\rightarrow \infty }P(|\hat{f}_{TPE}^{M,\hat{\boldsymbol{\theta}}%
_{M}}(x_{Max})-\hat{f}_{TPE}^{M,\boldsymbol{\theta}_{M}}(x_{Max})|<\delta
)=1,
\end{equation}%
so that
\begin{equation}
\lim_{K\rightarrow \infty }P(|\hat{f}_{TPE}^{M,\hat{\boldsymbol{\theta}}%
_{M}}(x_{Max})-\hat{f}_{TPE}^{M,\boldsymbol{\theta}%
_{M}}(x_{Max})-f(x_{Max})+f(x_{Max})|<\delta )=1.
\end{equation}%
Since
\begin{multline}
|\hat{f}_{TPE}^{M,\hat{\boldsymbol{\theta}}_{M}}(x_{Max})-\hat{f}_{TPE}^{M,%
\boldsymbol{\theta}_{M}}(x_{Max})-f(x_{Max})+f(x_{Max})|\geq \\
||\hat{f}_{TPE}^{M,\hat{\boldsymbol{\theta}}_{M}}(x_{Max})-f(x_{Max})|-|\hat{%
f}_{TPE}^{M,\boldsymbol{\theta}_{M}}(x_{Max})-f(x_{Max})||,  \notag
\end{multline}%
we have
\begin{multline}
\lim_{K\rightarrow \infty }P(||\hat{f}_{TPE}^{M,\hat{\boldsymbol{\theta}}%
_{M}}(x_{Max})-f(x_{Max})|-|\hat{f}_{TPE}^{M,\boldsymbol{\theta}%
_{M}}(x_{Max})-f(x_{Max})||<\delta ) \\
\geq \lim_{K\rightarrow \infty }P(|\hat{f}_{TPE}^{M,\hat{\boldsymbol{\theta}}%
_{M}}(x_{Max})-\hat{f}_{TPE}^{M,\boldsymbol{\theta}%
_{M}}(x_{Max})-f(x_{Max})+f(x_{Max})|<\delta )=1,  \notag
\end{multline}%
so that
\begin{equation}
\lim_{K\rightarrow \infty }P(||\hat{f}_{TPE}^{M,\hat{\boldsymbol{\theta}}%
_{M}}(x_{Max})-f(x_{Max})|-|\hat{f}_{TPE}^{M,\boldsymbol{\theta}%
_{M}}(x_{Max})-f(x_{Max})||<\delta )=1.
\end{equation}%
Thus we can write
\begin{multline}
\lim_{K\rightarrow \infty }P(-\delta +|\hat{f}_{TPE}^{M,\boldsymbol{\theta}%
_{M}}(x_{Max})-f(x_{Max})|<|\hat{f}_{TPE}^{M,\hat{\boldsymbol{\theta}}%
_{M}}(x_{Max})-f(x_{Max})|< \\
\delta +|\hat{f}_{TPE}^{M,\boldsymbol{\theta}_{M}}(x_{Max})-f(x_{Max})|)=1
\notag
\end{multline}%
and therefore, we have
\begin{equation}
\lim_{K\rightarrow \infty }P(|\hat{f}_{TPE}^{M,\hat{\boldsymbol{\theta}}%
_{M}}(x_{Max})-f(x_{Max})|<\delta +|\hat{f}_{TPE}^{M,\boldsymbol{\theta}%
_{M}}(x_{Max})-f(x_{Max})|)\geq  \label{DeltaEq1}
\end{equation}%
\begin{multline*}
\lim_{K\rightarrow \infty }P(-\delta +|\hat{f}_{TPE}^{M,\boldsymbol{\theta}%
_{M}}(x_{Max})-f(x_{Max})|<|\hat{f}_{TPE}^{M,\hat{\boldsymbol{\theta}}%
_{M}}(x_{Max})-f(x_{Max})|< \\
\delta +|\hat{f}_{TPE}^{M,\boldsymbol{\theta}_{M}}(x_{Max})-f(x_{Max})|)=1.
\end{multline*}%
Using equation (\ref{LimitDenseTPE}), we have%
\begin{equation*}
\lim_{M\rightarrow \infty }\hat{f}_{TPE}^{M,\boldsymbol{\theta}%
_{M}}(x_{Max})=f(x_{Max}),
\end{equation*}%
so that for any $\delta _{1}>0,$ there exists $M^{\prime }>0,$ such that for
any $M>M^{\prime }$, we have
\begin{equation}
|\hat{f}_{TPE}^{M,\boldsymbol{\theta}_{M}}(x_{Max})-f(x_{Max})|<\delta _{1},
\end{equation}%
and adding $\delta $ on both sides, we obtain%
\begin{equation}
\delta +|\hat{f}_{TPE}^{M,\boldsymbol{\theta}_{M}}(x_{Max})-f(x_{Max})|<%
\delta +\delta _{1}=\delta _{0}.
\end{equation}%
As a result, for any $M>M^{\prime }$, equation (\ref{DeltaEq1}) becomes
\begin{equation}
\begin{split}
1& =\lim_{K\rightarrow \infty }P(|\hat{f}_{TPE}^{M,\hat{\boldsymbol{\theta}}%
_{M}}(x_{Max})-f(x_{Max})|<\delta +|\hat{f}_{TPE}^{M,\boldsymbol{\theta}%
_{M}}(x_{Max})-f(x_{Max})|) \\
& \leq \lim_{K\rightarrow \infty }P(|\hat{f}_{TPE}^{M,\hat{%
\boldsymbol{\theta}}_{M}}(x_{Max})-f(x_{Max})|<\delta _{0}),
\end{split}%
\end{equation}%
and since $\delta _{0}$ is also arbitrarily chosen, we have
\begin{equation}
\lim_{M\rightarrow \infty }\lim_{K\rightarrow \infty }P(|\hat{f}_{TPE}^{M,%
\hat{\boldsymbol{\theta}}_{M}}(x_{Max})-f(x_{Max})|<\delta _{0})=1,
\end{equation}%
so that
\begin{equation}
\lim_{M\rightarrow \infty }\hat{f}_{TPE}^{M,\hat{\boldsymbol{\theta}}%
_{M}}(x_{Max})\overset{\mbox{p}}{\longrightarrow }f(x_{Max}).
\end{equation}%
Now using equation (\ref{XmaxDef}), we can write%
\begin{equation}
\lim_{M\rightarrow \infty }P(|\hat{f}_{TPE}^{M,\hat{\boldsymbol{\theta}}%
_{M}}(x_{Max})-f(x_{Max})|<\delta _{0})=\lim_{M\rightarrow \infty
}P(\sup_{x}|\hat{f}_{TPE}^{M,\hat{\boldsymbol{\theta}}_{M}}(x)-f(x)|<\delta
_{0})=1,
\end{equation}%
for arbitrary $\delta _{0}>0,$ and an appeal to continuity of probability
measure (Micheas, 2018, Theorem 4.13) yields%
\begin{equation*}
P(\lim_{M\rightarrow \infty }\sup_{x}|\hat{f}_{TPE}^{M,\hat{%
\boldsymbol{\theta}}_{M}}(x)-f(x)|=0)=1,
\end{equation*}%
i.e., $\hat{f}_{TPE}^{M,\hat{\boldsymbol{\theta}}_{M}}(x)$ converges
uniformly almost surely to the function $f(x)$, for all $x\in \Re $, as $%
K,M\rightarrow +\infty $. Now using equation (\ref{Dmbound}) we can write
\begin{equation*}
\lim_{M\rightarrow \infty }P\left( D_{M}^{L}<L\delta _{0}\right) \geq
\lim_{M\rightarrow \infty }P(L\sup_{x}|\hat{f}_{TPE}^{M,\hat{%
\boldsymbol{\theta}}_{M}}(x)-f(x)|<L\delta _{0})=1,
\end{equation*}%
and sending $L\rightarrow \infty ,$ with $L\delta _{0}\rightarrow 0,$ we have%
\begin{equation*}
\lim_{M\rightarrow \infty }P\left( D_{M}\rightarrow 0\right) =1,
\end{equation*}%
which leads to%
\begin{equation}
\hat{f}_{TPE}^{M,\hat{\boldsymbol{\theta}}_{M}}(x)\overset{\mathcal{L}^{1}}{%
\rightarrow }f(x),
\end{equation}%
as $K,M\rightarrow +\infty .$
\end{proof}

\subsection{Extension to the Multivariate Case\label{s3.2}}

Based on the results of the previous section, the extension to higher
dimensions is straightforward. In particular, suppose now that $f:\Re
^{d}\rightarrow \Re $, is analytic at the point $\mathbf{x}_{0}=$ $%
(x_{1,0},x_{2,0},$ $\dots ,x_{d,0})\in \Re ^{d},$ so that its Taylor
expansion is given by
\begin{equation*}
f(\mathbf{x})=\sum_{n_{1}=0}^{+\infty }\sum_{n_{2}=0}^{+\infty }\dots
\sum_{n_{d}=0}^{+\infty }a_{n_{1},n_{2},\dots ,n_{d}}
(x_{1}-x_{1,0})^{n_{1}}(x_{2}-x_{2,0})^{n_{2}}\dots (x_{d}-x_{d,0})^{n_{d}},
\end{equation*}%
where $\mathbf{x}=$ $(x_{1},x_{2},$ $\dots ,x_{d})\in \Re ^{d},$ and
\begin{equation*}
a_{n_{1},n_{2},\dots ,n_{d}}=\frac{1}{n_{1}!n_{2}!\dots n_{d}!}\left( \frac{%
\partial ^{n_{1}+n_{2}+\dots +n_{d}}f}{\partial x_{1}^{n_{1}}\partial
x_{2}^{n_{2}}\dots \partial x_{d}^{n_{d}}}\right) (\mathbf{x}_{0}),
\end{equation*}%
denotes all the coefficients. Thus, the truncated Taylor series with $M$
terms is given by
\begin{equation}
\widehat{f}(\mathbf{x})=\sum_{n_{1}=0}^{M_{1}-1}\sum_{n_{2}=0}^{M_{2}-1}%
\dots
\sum_{n_{d}=0}^{M_{d}-1}a_{n_{1},n_{2},...,n_{d}}(x_{1}-x_{1,0})^{n_{1}}(x_{2}-x_{2,0})^{n_{2}}\dots (x_{d}-x_{d,0})^{n_{d}},
\end{equation}%
where $M_{1}M_{2}\dots M_{d}=M$.

Similarly to the univariate case, consider a region $\mathcal{W}\subset \Re
^{d+1}$, and suppose that we observe $v$ events $\{\mathbf{s}%
_{k}\}_{j=1}^{v}=\{(a_{j},n_{1,j},n_{2,j},\dots ,n_{d,j})\}_{j=1}^{v}$,
where $(a_{j},n_{1,j},n_{2,j},\dots ,n_{d,j})\in $ $\mathcal{W}$. Then, we
define the Poisson point process $\mathcal{N}$ over the window $\mathcal{W}$%
, with mixture intensity function $\lambda (\mathbf{s})$, $\mathbf{s}%
=(a,n_{1},n_{2},\dots ,n_{d})\in \mathcal{W}$, given by
\begin{equation}
\lambda (a,n_{1},n_{2},\dots ,n_{d})=\lambda
\sum_{m=1}^{M}p_{m}g_{m}(a,n_{1},n_{2},\dots ,n_{d}|\mathbf{\mu }_{m},%
\mathbf{\Sigma }_{m}),  \label{MultIntesFull}
\end{equation}%
where $0\leq p_{m}\leq 1$, $p_{1}+p_{2}+\dots +p_{M}=1$,
\begin{equation}
\mathbf{\mu }_{m}=(%
\begin{array}{ccccc}
\mu _{a,m} & \mu _{n_{1},m} & \mu _{n_{2},m} & \dots & \mu _{n_{d},m}%
\end{array}%
)=(%
\begin{array}{cc}
\mu _{a,m} & \mathbf{\xi }_{m}%
\end{array}%
),  \label{MuSplit}
\end{equation}%
with $\mathbf{\xi }_{m}=(%
\begin{array}{cccc}
\mu _{n_{1},m} & \mu _{n_{2},m} & \dots & \mu _{n_{d},m}%
\end{array}%
),$ and in general,%
\begin{equation}
\mathbf{\Sigma }_{m}=%
\begin{pmatrix}
\sigma _{a,m}^{2} & \rho _{1,m}\sigma _{a,m}\sigma _{n_{1},m} & \rho
_{2,m}\sigma _{a,m}\sigma _{n_{2},m} & \dots & \rho _{d,m}\sigma
_{a,m}\sigma _{n_{d},m} \\
\rho _{1,m}\sigma _{a,m}\sigma _{n_{1},m} & \sigma _{n_{1},m}^{2} & \sigma
_{n_{1},n_{2},m} & \dots & \sigma _{n_{1},n_{d},m} \\
\rho _{2,m}\sigma _{a,m}\sigma _{n_{2},m} & \sigma _{n_{1},n_{2},m} & \sigma
_{n_{2},m}^{2} & \dots & \sigma _{n_{2},n_{d},m} \\
\dots & \dots & \dots & \dots & \dots \\
\rho _{d,m}\sigma _{a,m}\sigma _{n_{d},m} & \sigma _{n_{1},n_{d},m} & \sigma
_{n_{2},n_{d},m} & \dots & \sigma _{n_{d},m}^{2}%
\end{pmatrix}%
,
\end{equation}%
so that
\begin{equation*}
E[\mathcal{N}(B)]=\Lambda (B)=\int_{B}\lambda (\mathbf{s})\mu _{d+1}(d%
\mathbf{s})=M,
\end{equation*}%
where $\mu _{d+1}$ denotes Lebesgue measure in $\Re ^{d+1}$. Intuitively, it
makes sense for the powers of the Taylor expansion terms to be independent,
and therefore we will consider the diagonal structure below%
\begin{equation}
\mathbf{\Sigma }_{m}=%
\begin{pmatrix}
\sigma _{a,m}^{2} & \rho _{1,m}\sigma _{a,m}\sigma _{n_{1},m} & \rho
_{2,m}\sigma _{a,m}\sigma _{n_{2},m} & \dots & \rho _{d,m}\sigma
_{a,m}\sigma _{n_{d},m} \\
\rho _{1,m}\sigma _{a,m}\sigma _{n_{1},m} & \sigma _{n_{1},m}^{2} & 0 & \dots
& 0 \\
\rho _{2,m}\sigma _{a,m}\sigma _{n_{2},m} & 0 & \sigma _{n_{2},m}^{2} & \dots
& 0 \\
\dots & \dots & \dots & \dots & \dots \\
\rho _{d,m}\sigma _{a,m}\sigma _{n_{d},m} & 0 & 0 & \dots & \sigma
_{n_{d},m}^{2}%
\end{pmatrix}%
,
\end{equation}%
and in addition, we will write%
\begin{equation}
\mathbf{\Sigma }_{m}=%
\begin{pmatrix}
\sigma _{a,m}^{2} & \Sigma _{12}^{T} \\
\Sigma _{12} & \mathbf{\Sigma }_{22}%
\end{pmatrix}%
,  \label{CovMat}
\end{equation}%
where $\Sigma _{12}^{T}=(\rho _{1,m}\sigma _{a,m}\sigma _{n_{1},m},\dots
,\rho _{d,m}\sigma _{a,m}\sigma _{n_{d},m}),$ and $\mathbf{\Sigma }%
_{22}=diag(\sigma _{n_{1},m}^{2},\dots ,\sigma _{n_{d},m}^{2}).$ Similarly
to the univariate case, we set $\lambda =E[\mathcal{N}(B)]=M,$ and $p_{m}=%
\frac{1}{M}$ for all $m=1,2,\dots ,M$. As a result, the intensity function
becomes%
\begin{equation}
\lambda (a,n_{1},n_{2},\dots ,n_{d})=\sum_{m=1}^{M}g_{m}(a,n_{1},n_{2},\dots
,n_{d}|\mathbf{\mu }_{m},\mathbf{\Sigma }_{m}).
\label{MultivIntensityFunction}
\end{equation}

Now, given a realization from this IPPP, say, $\varphi
_{v}=\{(a_{1},n_{1,1},n_{2,1},\dots $, $n_{d,1})$, $(a_{2},$ $%
n_{1,2},n_{2,2},\dots ,n_{d,2})$, $\dots $, $(a_{v},n_{1,v},n_{2,v},\dots
,n_{d,v})\}$, and for any $\mathbf{x}$ which satisfies $x_{1}>x_{1,0},$ $%
x_{2}>x_{2,0},$ $\dots ,$ $x_{d}>x_{d,0}$, the function $f(\mathbf{x})$ is
constructed by
\begin{equation*}
\hat{f}(\mathbf{x})=%
\sum_{j=1}^{v}a_{j}(x_{1}-x_{1,0})^{n_{1,j}}(x_{2}-x_{2,0})^{n_{2,j}}\dots
(x_{d}-x_{0,d})^{n_{d,j}}.
\end{equation*}

The multivariate analog of Theorem (\ref{TaylorExpansionMeanThm}) is
presented next.

\begin{theorem}[Multivariate Taylor Expansion via IPPP]
\label{MultivTaylorExpansionMeanThm} Consider the random variables
\begin{equation}
\hat{f}_{\mathcal{N}}(\mathbf{x})=\sum_{(a,n_{1},n_{2},\dots ,n_{d})\in
\mathcal{N}}a(x_{1}-x_{1,0})^{n_{1}}(x_{2}-x_{2,0})^{n_{2}}\dots
(x_{d}-x_{d,0})^{n_{d}},
\end{equation}%
where $\mathcal{N}$ denotes an IPPP with the intensity surface of equation (%
\ref{MultIntesFull}), and assume that $x_{1}>x_{1,0},$ $x_{2}>x_{2,0},\dots
, $ $x_{d}>x_{d,0}$. Then the multivariate Taylor expansion Poisson point
process estimator (MTPE) of the function $f(\mathbf{x})$ is given by
\begin{equation}
\begin{split}
\hat{f}_{MTPE}^{M,\boldsymbol{\theta}_M}(\mathbf{x})= E(\hat{f}_{\mathcal{N}%
}(\mathbf{x}))=& \lambda \sum_{m=1}^{M}p_{m}\left( \mu
_{a,m}+\sum_{r=1}^{d}\rho _{r,m}\sigma _{a,m}\sigma _{n_{r},m}\ln
(x_{r}-x_{r,0})\right) \\
& \prod_{r=1}^{d}(x_{r}-x_{r,0})^{\mu _{n_{r},m}+\frac{\sigma
_{n_{_{r}},m}^{2}}{2}\ln (x_{r}-x_{r,0})}.
\end{split}
\label{e3.4}
\end{equation}
\end{theorem}

\begin{proof}
Similarly to the univariate case, let $\mathbf{n}=(n_{1},n_{2},\dots ,n_{d})$%
, $h(a,\mathbf{n})=a(x_{1}-x_{1,0})^{n_{1}}$ $\dots $ $%
(x_{d}-x_{d,0})^{n_{d}}$, $\mathbf{x}=$ $(x_{1},x_{2},$ $\dots ,x_{d}),$ $%
\mathbf{x}_{0}=$ $(x_{1,0},x_{2,0},$ $\dots ,x_{d,0}),$ and write
\begin{equation}
\hat{f}_{\mathcal{N}}(\mathbf{x})=\sum_{(a,\mathbf{n})\in \mathcal{N}%
}a(x_{1}-x_{1,0})^{n_{1}}\dots (x_{d}-x_{d,0})^{n_{d}}=\sum_{(a,\mathbf{n}%
)\in \mathcal{N}}h(a,\mathbf{n}).
\end{equation}%
An appeal to Campbell's theorem (\cite{r47}) for point process sums, yields
\begin{equation*}
\begin{split}
E(\hat{f}_{\mathcal{N}}(\mathbf{x}))& =E\left( \sum_{(a,\mathbf{n})\in
\mathcal{N}}h(a,\mathbf{n})\right) =\int_{\Re ^{d+1}}h(a,\mathbf{n})\lambda
(a,\mathbf{n})dad\mathbf{n} \\
& =\lambda \int_{\Re ^{d+1}}a(x_{1}-x_{1,0})^{n_{1}}\dots
(x_{d}-x_{d,0})^{n_{d}}\sum_{m=1}^{M}p_{m}g(a,\mathbf{n}|\mathbf{\mu }_{m},%
\mathbf{\Sigma }_{m})dad\mathbf{n}
\end{split}%
\end{equation*}%
so that%
\begin{equation}
E(\hat{f}_{\mathcal{N}}(\mathbf{x}))=\lambda \sum_{m=1}^{M}p_{m}\int_{\Re
^{d+1}}a(x_{1}-x_{1,0})^{n_{1}}\dots (x_{d}-x_{d,0})^{n_{d}}g_{m}(a,\mathbf{n%
}|\mathbf{\mu }_{m},\mathbf{\Sigma }_{m})dad\mathbf{n},
\label{MTPEproofMean}
\end{equation}%
and we require calculation of the integral above. Recall equations (\ref%
{MuSplit}) and (\ref{CovMat}), and write the joint multivariate normal
component distribution as%
\begin{eqnarray}
g_{m}(a,\mathbf{n}|\mathbf{\mu }_{m},\mathbf{\Sigma }_{m}) &=&g_{m}(\mathbf{n%
}|\mathbf{\xi }_{m},\mathbf{\Sigma }_{22})g_{m}(a|\mathbf{n},\mathbf{\mu }%
_{m},\mathbf{\Sigma }_{m})  \notag \\
&=&g_{m}(a|\mathbf{n},\mathbf{\mu }_{m},\mathbf{\Sigma }_{m})\prod%
\limits_{r=1}^{d}\phi (n_{r}|\mu _{n_{r},m},\sigma _{n_{r},m}^{2}),
\end{eqnarray}%
where%
\begin{equation*}
g_{m}(\mathbf{n}|\mathbf{\xi }_{m},\mathbf{\Sigma }_{22})=\prod%
\limits_{i=1}^{d}\phi (n_{r}|\mu _{n_{r},m},\sigma _{n_{r},m}^{2}),
\end{equation*}%
with $\phi (n_{r}|\mu _{n_{r},m},\sigma _{n_{r},m}^{2})$ the density of a
normal $N(\mu _{n_{r},m},\sigma _{n_{r},m}^{2})$. Then we can write%
\begin{equation*}
\begin{split}
I& =\int_{\Re ^{d+1}}a(x_{1}-x_{1,0})^{n_{1}}\dots
(x_{d}-x_{d,0})^{n_{d}}g_{m}(a,\mathbf{n}|\mathbf{\mu }_{m},\mathbf{\Sigma }%
_{m})dad\mathbf{n} \\
& =\int_{\Re ^{d}}(x_{1}-x_{1,0})^{n_{1}}\dots (x_{d}-x_{d,0})^{n_{d}}g_{m}(%
\mathbf{n}|\mathbf{\xi }_{m},\mathbf{\Sigma }_{22})\int_{\Re }ag_{m}(a|%
\mathbf{n},\mathbf{\mu }_{m},\mathbf{\Sigma }_{m})dad\mathbf{n} \\
& =\int_{\Re ^{d}}E^{g_{m}}(a|\mathbf{n},\mathbf{\mu }_{m},\mathbf{\Sigma }%
_{m})\prod\limits_{q=1}^{d}(x_{q}-x_{q,0})^{n_{q}}\phi (n_{q}|\mu
_{n_{q},m},\sigma _{n_{q},m}^{2})d\mathbf{n},
\end{split}%
\end{equation*}%
with $a|\mathbf{n},\mathbf{\mu }_{m},\mathbf{\Sigma }_{m}\backsim N(\mu
_{a,m}+\Sigma _{12}\Sigma _{22}^{-1}(\mathbf{n}-\mathbf{\xi }_{m}),\sigma
_{a,m}-\Sigma _{12}\Sigma _{22}^{-1}\Sigma _{21}),$ so that%
\begin{equation*}
E^{g_{m}}(a|\mathbf{n},\mathbf{\mu }_{m},\mathbf{\Sigma }_{m})=\mu
_{a,m}+\Sigma _{12}\Sigma _{22}^{-1}(\mathbf{n}-\mathbf{\xi }_{m})=\mu
_{a,m}+\sum_{r=1}^{d}\rho _{r,m}\frac{\sigma _{a,m}}{\sigma _{n_{r},m}}%
(n_{r}-\mu _{n_{r},m}).
\end{equation*}%
As a result, the integral $I$ becomes%
\begin{eqnarray*}
I &=&\int_{\Re ^{d}}\left[ \mu _{a,m}+\sigma _{a,m}\sum_{r=1}^{d}\frac{\rho
_{r,m}}{\sigma _{n_{r},m}}(n_{r}-\mu _{n_{r},m})\right] \prod%
\limits_{q=1}^{d}(x_{q}-x_{q,0})^{n_{q}}\phi (n_{q}|\mu _{n_{q},m},\sigma
_{n_{q},m}^{2})d\mathbf{n} \\
&=&\mu _{a,m}\prod\limits_{r=1}^{d}\int_{\Re }(x_{r}-x_{r,0})^{n_{r}}\phi
(n_{r}|\mu _{n_{r},m},\sigma _{n_{r},m}^{2})dn_{r} \\
&&+\sigma _{a,m}\sum_{r=1}^{d}\frac{\rho _{r,m}}{\sigma _{n_{r},m}}\left[
\int_{\Re ^{d}}(n_{r}-\mu
_{n_{r},m})\prod\limits_{q=1}^{d}(x_{q}-x_{q,0})^{n_{q}}\phi (n_{q}|\mu
_{n_{q},m},\sigma _{n_{q},m}^{2})d\mathbf{n}\right] ,
\end{eqnarray*}%
and therefore%
\begin{equation*}
I=\mu _{a,m}\prod\limits_{r=1}^{d}I_{r}+\sigma _{a,m}\sum_{r=1}^{d}\frac{%
\rho _{r,m}}{\sigma _{n_{r},m}}I_{r}^{\ast },
\end{equation*}%
where from equation (\ref{Exp1}) we have%
\begin{equation*}
I_{r}=\int_{\Re }(x_{r}-x_{r,0})^{n_{r}}\phi (n_{r}|\mu _{n_{r},m},\sigma
_{n_{r},m}^{2})dn_{r}=(x_{r}-x_{r,0})^{\mu _{n_{r},m}+\frac{1}{2}\sigma
_{n_{_{r}},m}^{2}\log (x_{r}-x_{r,0})},
\end{equation*}%
and%
\begin{equation*}
I_{r}^{\ast }=\int_{\Re ^{d}}(n_{r}-\mu
_{n_{r},m})\prod\limits_{q=1}^{d}(x_{q}-x_{q,0})^{n_{q}}\phi (n_{q}|\mu
_{n_{q},m},\sigma _{n_{q},m}^{2})d\mathbf{n}.
\end{equation*}%
Now write%
\begin{eqnarray*}
I_{r}^{\ast } &=&\left[ \int_{\Re }(n_{r}-\mu
_{n_{r},m})(x_{r}-x_{r,0})^{n_{r}}\phi (n_{r}|\mu _{n_{r},m},\sigma
_{n_{r},m}^{2})dn_{r}\right] \\
&&\prod\limits_{q=1,q\neq r}^{d}\int_{\Re }(x_{q}-x_{q,0})^{n_{q}}\phi
(n_{q}|\mu _{n_{q},m},\sigma _{n_{q},m}^{2})d\mathbf{n},
\end{eqnarray*}%
and using equations (\ref{Exp1}) and (\ref{Exp2}) yields%
\begin{eqnarray*}
I_{r}^{\ast } &=&\left[ \sigma _{n_{r},m}^{2}\log (x_{r}-x_{r,0})\right]
(x_{r}-x_{r,0})^{\mu _{n_{r},m}+\frac{1}{2}\sigma _{n_{r},m}^{2}\log
(x_{r}-x_{r,0})} \\
&&\prod\limits_{q=1,q\neq r}^{d}(x_{q}-x_{q,0})^{\mu _{n_{q},m}+\frac{1}{2}%
\sigma _{n_{_{q}},m}^{2}\log (x_{q}-x_{q,0})} \\
&=&\sigma _{n_{r},m}^{2}\log
(x_{r}-x_{r,0})\prod\limits_{q=1}^{d}(x_{q}-x_{q,0})^{\mu _{n_{q},m}+\frac{1%
}{2}\sigma _{n_{_{q}},m}^{2}\log (x_{q}-x_{q,0})}.
\end{eqnarray*}%
Thus we can write%
\begin{eqnarray*}
I &=&\mu _{a,m}\prod\limits_{r=1}^{d}(x_{r}-x_{r,0})^{\mu _{n_{r},m}+\frac{1%
}{2}\sigma _{n_{_{r}},m}^{2}\log (x_{r}-x_{r,0})} \\
&&+\sigma _{a,m}\sum_{r=1}^{d}\frac{\rho _{r,m}}{\sigma _{n_{r},m}}\sigma
_{n_{r},m}^{2}\log (x_{r}-x_{r,0})\prod\limits_{q=1}^{d}(x_{q}-x_{q,0})^{\mu
_{n_{q},m}+\frac{1}{2}\sigma _{n_{_{q}},m}^{2}\log (x_{q}-x_{q,0})} \\
&=&\left[ \mu _{a,m}+\sigma _{a,m}\sum_{r=1}^{d}\rho _{r,m}\sigma
_{n_{r},m}\log (x_{r}-x_{r,0})\right] \prod%
\limits_{r=1}^{d}(x_{r}-x_{r,0})^{\mu _{n_{r},m}+\frac{1}{2}\sigma
_{n_{_{r}},m}^{2}\log (x_{r}-x_{r,0})},
\end{eqnarray*}%
as entertained.
\end{proof}

Similarly to the univariate case, we can obtain all the results of the
previous section for the MTPE (omitted). In particular, the modified version
of the MTPE is given by
\begin{equation}
\begin{split}
\hat{f}_{MTPE}^{M,\boldsymbol{\theta}_{M}}(\mathbf{x})=&
\sum_{m=1}^{M}\left( \mu _{a,m}+\sum_{r=1}^{d}\rho _{r,m}\sigma _{a,m}\sigma
_{n_{r},m}\ln (x_{r}-x_{r,0})\right) \\
& \prod_{r=1}^{d}(x_{r}-x_{r,0})^{\mu _{n_{r},m}+\frac{\sigma
_{n_{_{r}},m}^{2}}{2}\ln (x_{r}-x_{r,0})}.
\end{split}
\label{e3.5}
\end{equation}%
Next we put the theoretical results to use and discuss recovering a function
based on observed data.

\section{Implementation and simulation study}

In this section we consider the backward direction of the stochastic Taylor
expansion as follows: given observed inputs and outputs from a function,
estimate the coefficients of the underlying Poisson point process model. As
a result, we provide an estimator for the function itself within the range
of the observed inputs, and more importantly, we are able to perform
function extrapolation. We begin by discussing the algorithm required for
function estimation, followed by illustrative examples in order to study the
behavior of the proposed methodology in different scenarios.

\subsection{Algorithm: Recovering the function from data\label{s3.1}}

Suppose we have data $\boldsymbol{X}=\{\boldsymbol{x}_{k}\}_{k=1}^{K}$,
where $\boldsymbol{x}_{k}=(x_{1,k},x_{2,k},...,x_{d,k})$, with corresponding
values $\boldsymbol{y}=(y_{1},y_{2},...,y_{K})$, i.e., observed data from a
function $f:\Re ^{d}\rightarrow \Re $. Consider the non-linear regression
model%
\begin{equation}
y_{k}=\sum_{m=1}^{M}\left( \mu _{a,m}+\sum_{r=1}^{d}\rho _{r,m}\sigma
_{a,m}\sigma _{n_{r},m}\ln (x_{r,k}-x_{r,0})\right)
\prod_{r=1}^{d}(x_{r,k}-x_{r,0})^{\mu _{n_{r},m}+\frac{\sigma _{n_{r},m}^{2}%
}{2}\ln (x_{r,k}-x_{r,0})}+\epsilon _{k},
\end{equation}%
where $\epsilon _{k}\overset{iid}{\sim }N(0,\sigma ^{2})$, $k=1,2,...,K$.

The following Algorithm presents the steps in order to estimate the
parameters of the underlying IPPP model.

~\newline
\textbf{Algorithm 1}

\textbf{Step 1}: Set $\boldsymbol{x}_0=(\underset{k}{\min} (x_{1,k}),%
\underset{k}{\min}(x_{2,k}),\dots, \underset{k}{\min}(x_{d,k}))$, and choose
a maximum number of terms $M_{\max}$.

\textbf{Step 2}:

For each $M=1,2,\dots,M_{\max}$, fit the non-linear regression model, obtain
the MLE $\hat{\boldsymbol{\theta}}_M$ of $\boldsymbol{\theta}_M$ (same as
the least squares estimator), and calculate the MTPE $\hat{y}_k=\hat{f}%
_{MTPE}^ {M,\hat{\boldsymbol{\theta}}_M} (\mathbf{x}_k)$ using equation (\ref%
{e3.5}).

\textbf{Step 3}:

For each $M=1,2,\dots,M_{\max}$, calculate the residual sum of squares
\begin{equation}
RSS_M=\sum_{k=1}^K(y_k-\hat{f}^{M,\hat{\boldsymbol {\theta}}_M}_{TPE}(%
\boldsymbol{x}_k))^2,
\end{equation}
and choose the optimum number of terms $M$ that gives the smallest $RSS_M$,
i.e.,
\begin{equation}
\hat{M}=\arg\min(RSS_M).  \label{e4.2}
\end{equation}
The best MTPE is then given by $\hat{f}_{MTPE}^{\hat{M},\hat{%
\boldsymbol{\theta}}_{\hat{M}}}$.

Next we present several illustrative simulations in order to appreciate the
behavior of the proposed estimators and assess their performance.

\subsection{Simulations}

We conduct simulations using known functions, and compare the estimators
given by our algorithm against the truth. In particular, since we know the
true function, as a measure of overall performance we will calculate the
integrated distance between our final estimator $\hat{f}(\boldsymbol{x})$,
where $\boldsymbol{x}=(x_1,x_2,...,x_d)$, and the true function $f(%
\boldsymbol{x})$ over a certain set $W\subset\Re^ d$. The distance is given
by
\begin{equation}
D(\hat{f},f)=\int_W\left(\hat{f}(\boldsymbol{x})-f(\boldsymbol{x})\right)^2d%
\boldsymbol{x}.  \label{e4.1}
\end{equation}

In each of the following simulations, we consider a known function $f$ over
a given set $W$, and we draw a sample of size $K$ from $y_{k}=f(%
\boldsymbol{x}_{k})+\epsilon _{k}$, with $\epsilon _{k}\overset{iid}{\sim }%
N(0,\sigma ^{2})$, for some given $\sigma >0$, for all $k=1,2,...,K$. The
locations $\boldsymbol{x}_{k}$ are drawn uniformly and ordered, and then
used as the inputs to the function $f$. These samples $(\boldsymbol{x}%
_{k},y_{k})$ are then used in our algorithm to provide the TPE or MTPE for
the function, which allows us to extrapolate the function, as well as assess
the accuracy of the estimator by calculating the integrated distance of the
estimator and the truth using formula (\ref{e4.1}). All programming and
calculations where performed using R software, version 4.2.1.

\textbf{Univariate Examples}: In order to explore the behavior of the
estimator given by Algorithm 1, we perform simulation studies by considering
different functions, for several sample sizes. In all the examples that
follow, interpolation is near perfect which is a good indication that the
methodology is verified. However, as anticipated, during extrapolation and
for certain functions, we will observe departure from the truth the further
away we get from the data. This is standard behavior for any statistical
model when it comes to forecasting. From our simulations we have concluded
that this phenomenon occurs when first, the sample size is small and second,
the function under investigation is analytic and requires an infinite term
Taylor series.
\begin{figure}[htbp]
\centering \includegraphics[width=0.6\textwidth]{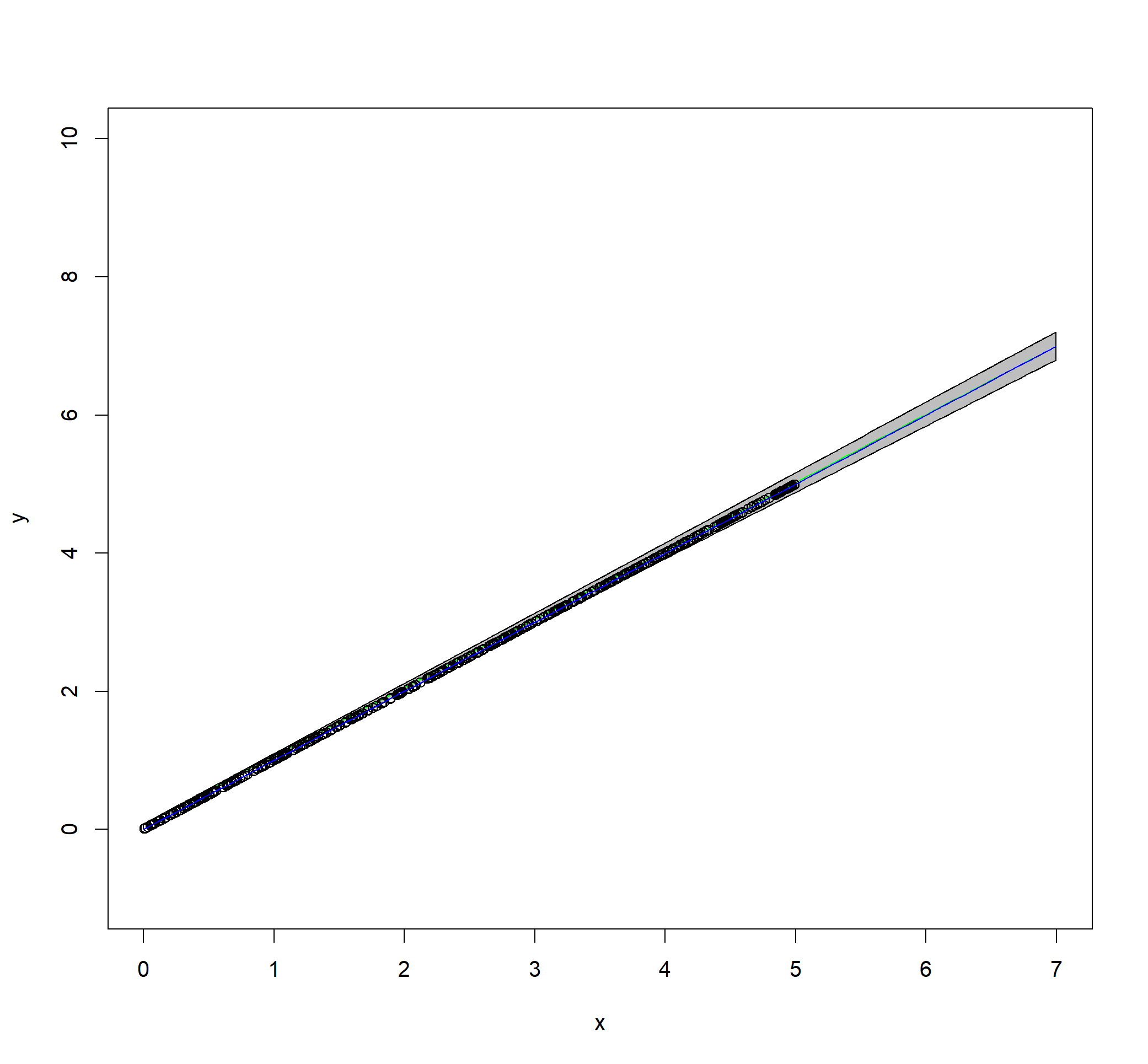}
\caption{Ideal function $f(x)=x$, with a sample of size $K=500$. As
expected, the estimated number of terms is $\hat{M}=1$. The TPE (red), the
mean of point process realizations (green), truth (blue), are all on top of
each other. The 95\% envelopes (grey) are also displayed.}
\label{straightline}
\end{figure}

As an ideal case, consider $f(x)=x$, with $x\in (0,5]$, $x_{0}=0$, $\sigma
=0.00001$, $M_{\max }=15$, and sample $K=500$ points. As expected, the TPE
over $x\in (0,7]$, yields an estimated number of terms of $\hat{M}=1$, with $%
D(\hat{f}_{500},f)=1.55161\cdot 10^{-10}$. The TPE is given by
\begin{equation}
\hat{f}_{500}(x)={\small (1.000023+2.031777\cdot 10^{-05}\ln
(x))x^{0.9999452+0.5\ln (x)2.517725\cdot 10^{-05}}.}
\end{equation}%
Notice how the estimates of the parameters adjust in order to give us a
complete recovery of the true function, since the $\log (x)$ function has no
contribution to the estimator (coefficients are near zero). In Figure \ref%
{straightline}, we display the TPE (red), the mean of point process
realizations (green), truth (blue), which are all on top of each other. The
95\% envelopes (grey) are also displayed. These bounds correspond to the $%
0.025$ and $0.975$ percentiles of 10000 STE realizations from the estimated
IPPP, using equation (\ref{TaylorPolyPPPEst}). All envelopes we present for
the examples that follow, are obtained in a similar fashion.

\begin{table}[htbp]
\centering
\begin{tabularx}{\textwidth}{XXXXXXXXX}
\hline
\raisebox{-0\height}[0.03\textwidth]{Sample} & $\hat M$ & $\boldsymbol{\hat\mu}_a$ & $\boldsymbol{\hat\mu}_n$ & $\boldsymbol{\hat\sigma}_a$ & $\boldsymbol{\hat\sigma}_n$ & $\boldsymbol{\hat\rho}$  & $D$ & Run Time \\ \hline
25 & 5 & $\tiny\begin{pmatrix}-3.770\\ -3.502\\ -1.248\\ 0.332\\ 3.214\end{pmatrix}$ & \tiny$\begin{pmatrix}0.549\\ 2.763\\ 5.778\\ 6.558\\ 4.294\end{pmatrix}$ & \tiny$\begin{pmatrix}
0.098\\ 0.1007\\ 0.096\\0.0895\\0.0991\end{pmatrix}$ & \tiny$\begin{pmatrix}0.348\\ 2.189\cdot 10^{-05}\\ 0.0191\\ 1.018\cdot 10^{-07}\\ 4.059\cdot 10^{-05}\end{pmatrix}$ & \tiny$\begin{pmatrix}
0.016\\ 0.0204\\ -0.028\\-0.0226\\ 0.0087\end{pmatrix}$ & 2283.702 & 2.92 s
 \\ \hline
\multicolumn{9}{c}{Estimator} \\
\hline
\multicolumn{9}{l}{$\hat{f}(x)=(-3.77+0.00056\ln(x)) x^{0.549+0.0606\ln(x)}+(-3.502+4.511\cdot 10^{-08}\ln(x))$}\\
\multicolumn{9}{c}{$x^{2.763+2.396\cdot 10^{-10} \ln(x)}+(-1.248-5.347\cdot 10^{-05}\ln(x))x^{5.778+1.831\cdot 10^{-04}\ln(x)}$}\\
\multicolumn{9}{c}{$+(0.332-2.061\cdot^{-10} \ln(x))x^{6.558+0.518\cdot 10^{-14}\ln(x)}$}\\
\multicolumn{9}{c}{$+(3.214+3.515\cdot 10^{-08}\ln(x))x^{4.294+0.824\cdot 10^{-09}\ln(x)}$}\\
\hline
\raisebox{-0\height}[0.03\textwidth]{Sample} & $\hat M$ & $\boldsymbol{\hat\mu}_a$ & $\boldsymbol{\hat\mu}_n$ & $\boldsymbol{\hat\sigma}_a$ & $\boldsymbol{\hat\sigma}_n$ & $\boldsymbol{\hat\rho}$  & $D$ & Run Time \\ \hline
100 & 5 & \tiny$\begin{pmatrix}-4.856\\-0.290\\0.385\\-0.127\\-0.067\end{pmatrix}$ & \tiny $\begin{pmatrix}1.372\\1.460\\1.6228\\0.543\\0.771\end{pmatrix}$ & \tiny $\begin{pmatrix}0.860\\5.834\\9.733\\12.422\\2.378\end{pmatrix}$ & \tiny $\begin{pmatrix}0.008\\1.046\\0.888\\0.490\\0.924\end{pmatrix}$ & \tiny$\begin{pmatrix}-0.289\\-0.999\\-0.272\\0.895\end{pmatrix}$ & 0.125 & 35.84 s
 \\ \hline
\multicolumn{9}{c}{Estimator} \\
\hline
\multicolumn{9}{l}{$\hat{f}(x)=(-4.856-0.002\ln(x))x^{1.372+0.004\ln(x)}+(-0.290-6.095\ln(x))x^{1.460+0.523\ln(x)}$}\\
\multicolumn{9}{c}{$+(0.385+8.586\ln(x))x^{1.628+0.444\ln(x)}+(-0.127-1.657\ln(x))x^{0.523+0.245\ln(x)}$}\\
\multicolumn{9}{c}{$+(-0.067+1.966\ln(x))x^{0.771+0.462\ln(x)}$}\\
\hline
\raisebox{-0\height}[0.03\textwidth]{Sample} & $\hat M$ & $\boldsymbol{\hat\mu}_a$ & $\boldsymbol{\hat\mu}_n$ & $\boldsymbol{\hat\sigma}_a$ & $\boldsymbol{\hat\sigma}_n$ & $\boldsymbol{\hat\rho}$  & $D$ & Run Time \\ \hline
500 & 4 & \tiny $\begin{pmatrix}-0.568\\-4.774\\-0.084\\0.468\end{pmatrix}$ & \tiny$\begin{pmatrix}0.841\\1.133\\1.458\\3.278\end{pmatrix}$ & \tiny$\begin{pmatrix}0.767\\2.751\\7.046\\0.849\end{pmatrix}$ & \tiny$\begin{pmatrix}0.336\\0.060\\0.233\\0.090\end{pmatrix}$ & \tiny$\begin{pmatrix}-0.9999\\0.381\\0.999997\\0.559\end{pmatrix}$ & 0.121 & 12.37 s
 \\ \hline
\multicolumn{9}{c}{Estimator} \\
\hline
\multicolumn{9}{l}{$\hat{f}(x)=(-0.568-0.258\ln(x))x^{0.841+0.168\ln(x)}+(-4.774+0.063\ln(x))x^{1.133+0.030\ln(x)}$}\\
\multicolumn{9}{c}{$+(-0.084+1.643\ln(x))x^{1.458+0.117\ln(x)}+(0.468+0.043\ln(x))x^{3.278+0.045\ln(x)}$}\\
\hline
\end{tabularx}
\caption{Simulation study for $f(x)=x^3-6x$, with $x\in(0,4]$, $x_0=0$, and $%
\protect\sigma=1$. We draw $K=25$, $100$, and $500$, samples from the true
function and set $M_{\max}=5$.}
\label{t5.1}
\end{table}

\begin{figure}[tbp]
\includegraphics[width=0.32\textwidth]{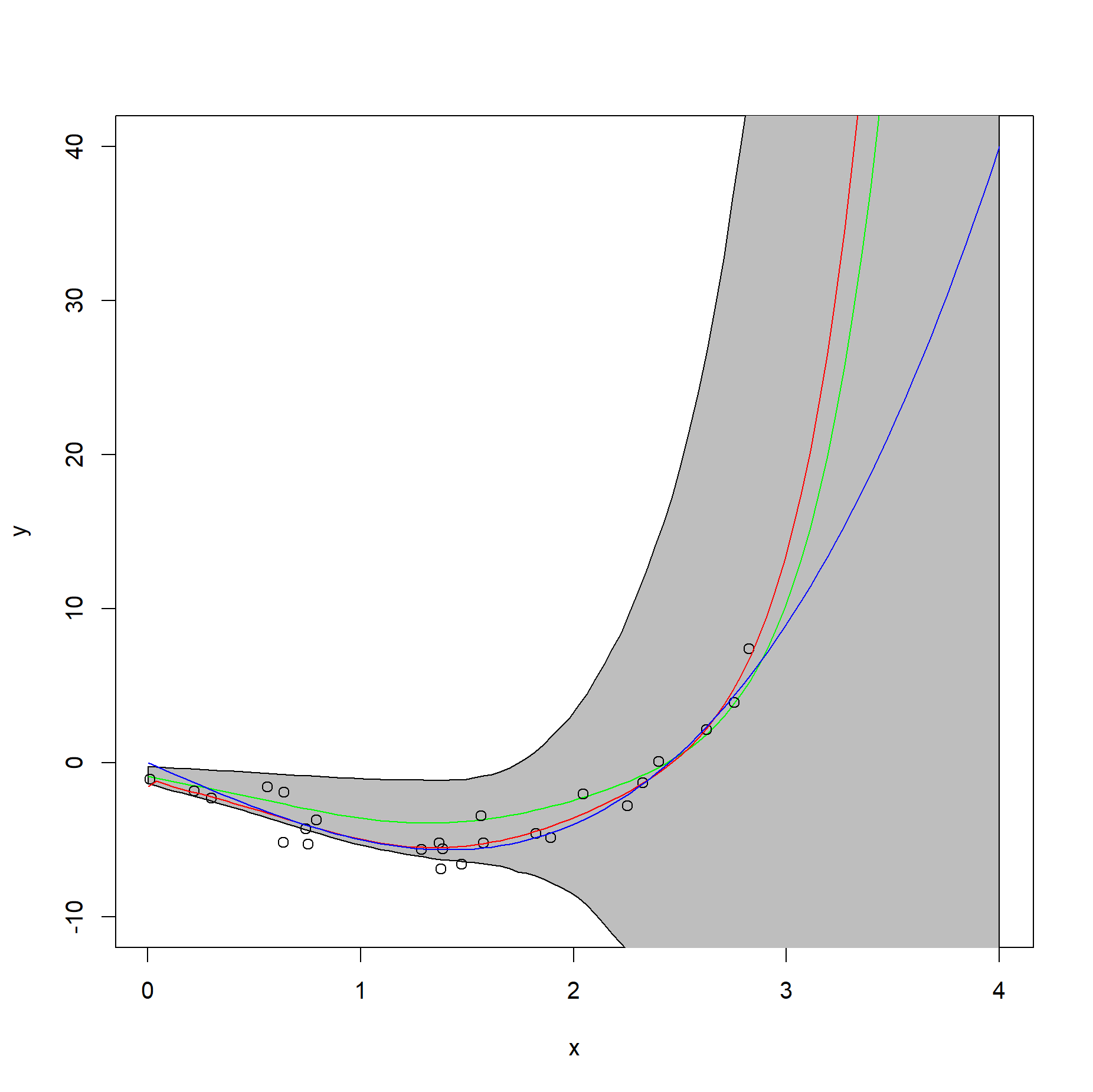} %
\includegraphics[width=0.32\textwidth]{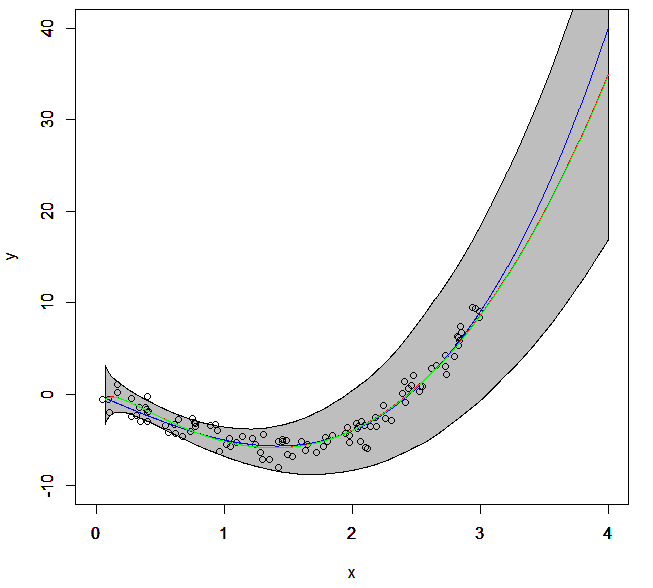} %
\includegraphics[width=0.32\textwidth]{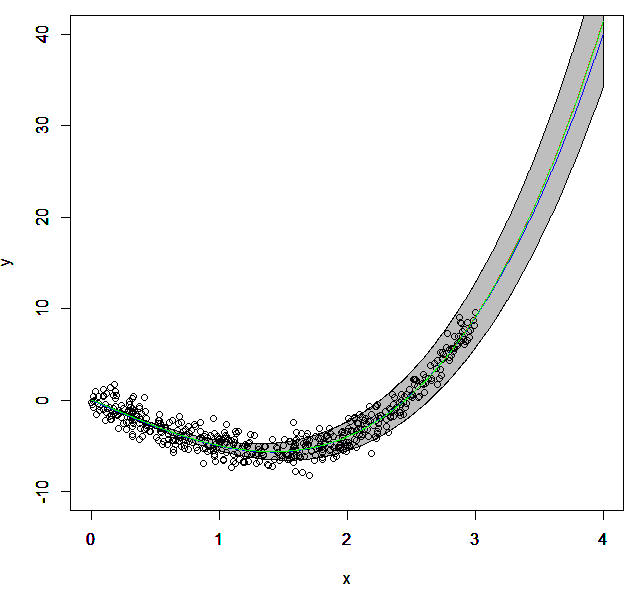}
\caption{Simulation study for $f(x)=x^3-6x$, with $x\in(0,4]$, $x_0=0$, and $%
\protect\sigma=1$. We draw $K=25$ (left), $100$ (middle), and $500$ (right),
samples from the true function and set $M_{\max}=5$. The TPE (red), the mean
of point process realizations (green), truth (blue), and the 95\% envelopes
(grey) are also displayed. }
\label{f5.1}
\end{figure}

Now we consider the function $f(x)=x^{3}-6x$, with $x\in (0,3]$, and set $%
x_{0}=0$, and $\sigma =1$. Once we obtain the TPE we perform extrapolation
in order to assess its performance in the interval $(0,4]$. We choose 3
different sample sizes of $K=25$, $100$, and $500$, and we set $M_{\max }=5$%
. Results are shown in Table \ref{t5.1} and we display the truth and fits in
Figure \ref{f5.1}. Notice how the estimator works in this case. Since we do
not have a perfect case (large $\sigma $), the final estimators consist of
more than two terms, with their estimated coefficients adjusting to give us
near perfect fits within the data. In terms of forecasting the truth in the
interval $(3,4]$, we can clearly see that the larger the sample size the
better the TPE fits, and this is also reflected in the integrated distances.

Next we consider a more complicated function that requires an infinite
number of terms in its Taylor series. Let $f(x)=x\sin(x)+e^{-x^2}+x%
\cos(x)/(x^2+1)$, over the interval $(0,3]$, and set $x_0=0, \sigma=0.2$. We
will assess the performance of the TPE in the interval $(0,4]$. We choose
three different sample sizes $K=25$, $100$, and $500$, and set $M_{\mbox{max}%
}=6$. Results are shown in Table \ref{t5.2} and Figure \ref{f5.2}. As
expected, we notice that the further away we get from the data
(extrapolation), the worst the estimated value of the function becomes,
however, the $95\%$ envelopes contain the truth.

\begin{table}[tbp]
\centering
\begin{tabularx}{\textwidth}{XXXXXXXXX}
\hline
\raisebox{-0\height}[0.03\textwidth]{Sample} &$\hat M$ & $\boldsymbol{\hat\mu}_a$ & $\boldsymbol{\hat\mu}_n$ & $\boldsymbol{\hat\sigma}_a$ & $\boldsymbol{\hat\sigma}_n$ & $\boldsymbol{\hat\rho}$  & $D$ & Run Time \\ \hline
25 & 6 & \tiny$\begin{pmatrix}0.733\\0.550\\0.177\\-0.078\\-0.010\\-0.001\end{pmatrix}$ & \tiny $\begin{pmatrix}-0.220\\0.586\\1.861\\2.558\\3.828\\4.499\end{pmatrix}$ & \tiny $\begin{pmatrix}1.999\\0.905\\1.197\\1.115\\1.066\\0.975\end{pmatrix}$ & \tiny $\begin{pmatrix}0.130\\0.387\\0.314\\0.308\\0.344\\0.319\end{pmatrix}$ & \tiny $\begin{pmatrix}0.168\\0.218\\0.333\\-0.055\\-0.059\\-0.006\end{pmatrix}$ & 1.048 & 202.4 s
 \\ \hline
\multicolumn{9}{c}{Estimator} \\
\hline
\multicolumn{9}{l}{$\hat{f}(x)=(0.733+0.044\ln(x))x^{-0.220+0.065\ln(x)}+(0.550+0.076\ln(x))x^{0.586+0.193\ln(x)}$}\\
\multicolumn{9}{c}{$+(0.177+0.125\ln(x))x^{1.861+0.157\ln(x)}+(-0.078-0.019\ln(x))x^{2.558+0.154\ln(x)}$}\\
\multicolumn{9}{c}{$+(-0.010-0.022\ln(x))x^{3.828+0.172\ln(x)}+(-0.001-0.002\ln(x))x^{4.499+0.159\ln(x)}$}\\
\hline
\raisebox{-0\height}[0.03\textwidth]{Sample} & $\hat M$ & $\boldsymbol{\hat\mu}_a$ & $\boldsymbol{\hat\mu}_n$ & $\boldsymbol{\hat\sigma}_a$ & $\boldsymbol{\hat\sigma}_n$ & $\boldsymbol{\hat\rho}$  & $D$ & Run Time \\ \hline
100 & 6 & \tiny$\begin{pmatrix}0.187\\-0.033\\  1.085\\  0.389\\ -0.194\\0.012
\end{pmatrix}$ & \tiny$\begin{pmatrix}0.400\\ 0.809\\ 1.503\\ 4.113\\ 3.610\\ 3.375\end{pmatrix}$ &\tiny $\begin{pmatrix}1.108\\ 15.056 \\ 0.454 \\ 0.214\\  3.278\\  0.657
\end{pmatrix}$ &\tiny $\begin{pmatrix}0.256\\ 0.335 \\0.205\\ 0.232 \\0.280 \\0.539
\end{pmatrix}$ & \tiny$\begin{pmatrix}-0.637\\ -0.345 \\ 0.873 \\ 0.147 \\-0.543\\ 0.090
\end{pmatrix}$ & 0.848 & 210.48 s
 \\ \hline
\multicolumn{9}{c}{Estimator} \\
\hline
\multicolumn{9}{l}{$\hat{f}(x)=(0.187-0.181\ln(x))x^{0.400+0.128\ln(x)}+(-0.033-1.742\ln(x))x^{0.809+0.168\ln(x)}$}\\
\multicolumn{9}{c}{$+(1.085+0.081\ln(x))x^{1.503+0.103\ln(x)}+(0.389+0.007\ln(x))x^{4.113+0.116\ln(x)}$}\\
\multicolumn{9}{c}{$+(-0.194-0.498\ln(x))x^{3.610+0.140\ln(x)}+(0.012+0.032\ln(x))x^{3.375+0.270\ln(x)}$}\\
\hline
\raisebox{-0\height}[0.03\textwidth]{Sample} &$\hat M$ & $\boldsymbol{\hat\mu}_a$ & $\boldsymbol{\hat\mu}_n$ & $\boldsymbol{\hat\sigma}_a$ & $\boldsymbol{\hat\sigma}_n$ & $\boldsymbol{\hat\rho}$  & $D$ & Run Time \\ \hline
500 & 4 & \tiny$\begin{pmatrix}0.459\\ 0.416\\0.636\\ 0.005
\end{pmatrix}$ & \tiny$\begin{pmatrix} 0.008\\0.840\\1.719\\ 3.301
\end{pmatrix}$ & \tiny$\begin{pmatrix}0.683\\ 8.435\\ 1.767\\ 0.753
\end{pmatrix}$ & \tiny$\begin{pmatrix}0.194\\ 0.170\\ 0.296\\ 0.244\end{pmatrix}$ &\tiny $\begin{pmatrix}-0.174 \\-0.903 \\ 0.439\\ -0.528\end{pmatrix}$  & 0.422 & 149.1 s
 \\ \hline
\multicolumn{9}{c}{Estimator} \\
\hline
\multicolumn{9}{l}{$\hat{f}(x)=(0.459-0.023\ln(x))x^{0.008+0.097\ln(x)}+(0.416-1.297\ln(x))x^{0.840+0.085\ln(x)}$}\\
\multicolumn{9}{c}{$+(0.636+0.230\ln(x))x^{1.719+0.148\ln(x)}+(0.005-0.097\ln(x))x^{3.301+0.122\ln(x)}$}\\
\hline
\end{tabularx}
\caption{Simulation study for $f(x)=x\sin(x)+e^{-x^2}+x\cos(x)/(x^2+1)$,
over the interval $x\in(0,4]$, with $x_0=0$, and $\protect\sigma=0.2$. We
draw $K=25$, $100$, and $500$, samples from the true function and set $%
M_{\max}=6$.}
\label{t5.2}
\end{table}

\begin{figure}[tbp]
\includegraphics[width=0.32\textwidth]{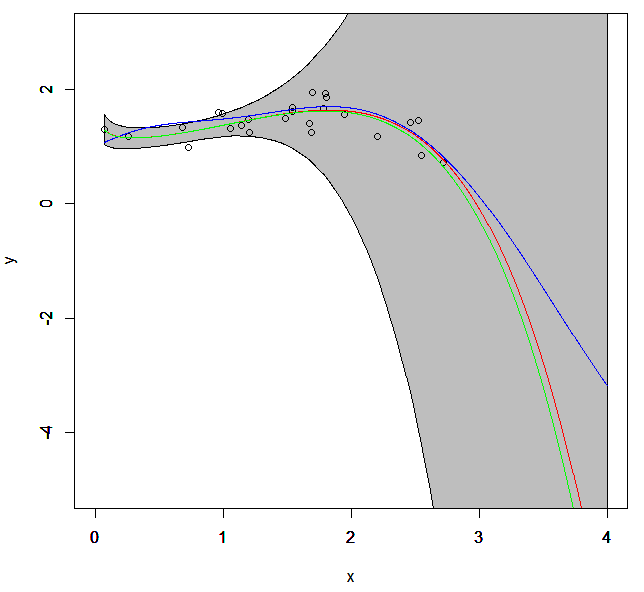} %
\includegraphics[width=0.32\textwidth]{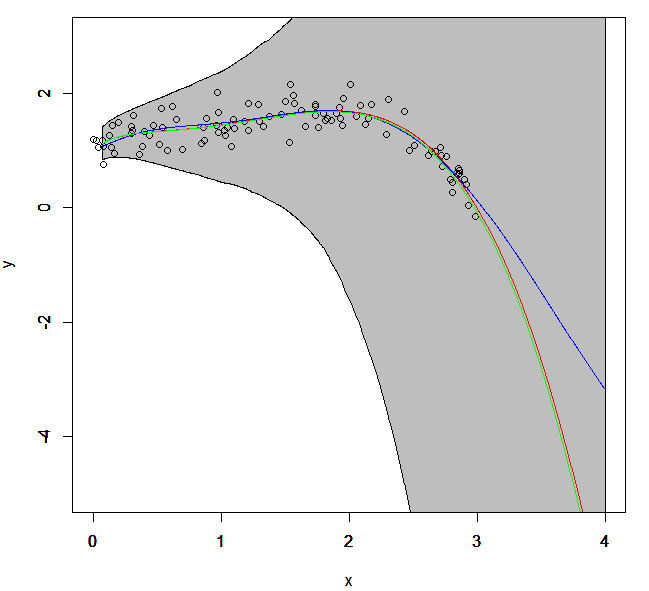} %
\includegraphics[width=0.32\textwidth]{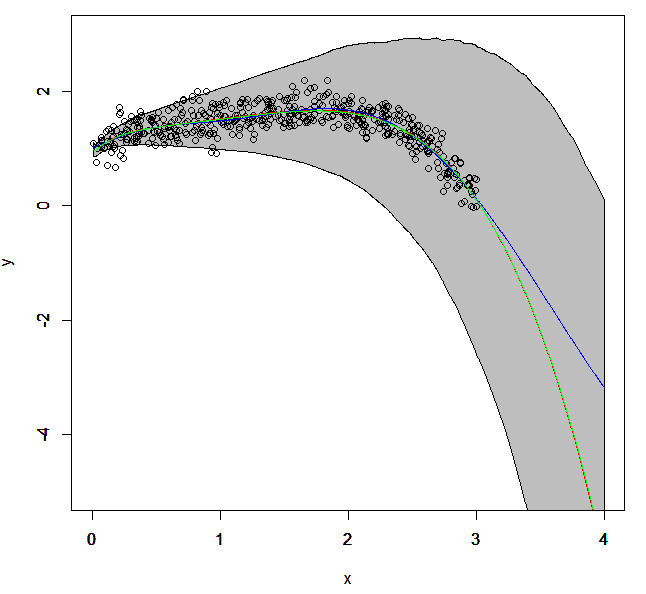}
\caption{Simulation study for $x\sin(x)+e^{-x^2}+x\cos(x)/(x^2+1)$, with $%
x\in(0,4]$, $x_0=0$, and $\protect\sigma=0.2$. We draw $K=25$ (left), $100$
(middle), and $500$ (right), samples from the true function and set $%
M_{\max}=6$. The TPE (red), the mean of point process realizations (green),
truth (blue), and the 95\% envelopes (grey) are also displayed.}
\label{f5.2}
\end{figure}

\textbf{Multivariate Examples}: Similarly to the univariate case, we tested
Algorithm 1 for several multivariate, real-valued functions. However, due to
the difficulty of plotting higher dimensional functions, only two
dimensional functions are presented. Moreover, since the number of
parameters is greatly increased, small samples sizes like 25 or 100, are not
enough in order to provide the MTPE. Therefore, we present results for $%
K=500 $.

\begin{table}[tbp]
\centering
\begin{tabularx}{\textwidth}{XXXXXX}
\hline
\raisebox{-0\height}[0.03\textwidth]{Sample} &$\hat M$ & $\boldsymbol{\hat\mu}_a$ & $\boldsymbol{\hat\mu}_{n_x}$ & $\boldsymbol{\hat\mu}_{n_y}$ & $\boldsymbol{\hat\sigma}_a$\\
\hline
500 & 6 & \tiny $\begin{pmatrix}1.514\\1.760\\-2.162\\-0.576\\-0.048\\0.430\end{pmatrix}$ & \tiny $\begin{pmatrix}0.150\\ 0.120\\ 1.354\\ 2.223\\ 6.282\\ 6.046\end{pmatrix}$ & \tiny $\begin{pmatrix}-0.086\\1.807\\0.818\\4.711\\3.436\\3.357\end{pmatrix}$ & \tiny $\begin{pmatrix}1.477\\0.987\\0.687\\1.220\\1.847\\ 0.402\end{pmatrix}$\\
\hline
$D$ & Run Time &$\boldsymbol{\hat\sigma}_{n_x}$ & $\boldsymbol{\hat\sigma}_{n_y}$ & $\boldsymbol{\hat\rho}_x$ & $\boldsymbol{\hat\rho}_y$   \\
\hline
0.032 & 96.71 s & \tiny  $\begin{pmatrix}0.171\\0.122\\0.090\\0.618\\0.155\\1.957\end{pmatrix}$& \tiny  $\begin{pmatrix}0.312\\0.199\\0.258\\0.151\\1.166\\ 0.970\end{pmatrix}$& \tiny  $\begin{pmatrix}-0.562\\-0.107\\-0.943\\-0.995\\-0.476\\0.923\end{pmatrix}$& \tiny  $\begin{pmatrix}0.910\\0.983\\0.452\\-0.964 \\-0.642\\0.370\end{pmatrix}$ \\
\hline
\multicolumn{6}{c}{Estimator} \\
\hline
\multicolumn{6}{l}{\tiny$\hat{f}(x,y)=$}\\
\multicolumn{6}{l}{\tiny$[1.514-0.142\ln(x+0.05)+0.419\ln(y+0.05)] (x+0.05)^{0.150+0.085\ln(x+0.05)} (y+0.05)^{-0.086+0.156\ln(y+0.05)}$}\\
\multicolumn{6}{l}{\tiny$+[1.760-0.013\ln(x+0.05)+0.193\ln(y+0.05)] (x+0.05)^{0.120+0.061\ln(x+0.05)}(y+0.05)^{1.807+0.100 \ln(y+0.05)}$}\\
\multicolumn{6}{l}{\tiny$+[-2.162-0.058\ln(x+0.05)+0.080\ln(y+0.05)] (x+0.05)^{1.354+0.045\ln(x+0.05)}(y+0.05)^{0.818+0.129 \ln(y+0.05)}$}\\
\multicolumn{6}{l}{\tiny$+[-0.576-0.749\ln(x+0.05)-0.177\ln(y+0.05)] (x+0.05)^{2.223+0.309\ln(x+0.05)}(y+0.05)^{4.711+0.075\ln(y+0.05)}$}\\
\multicolumn{6}{l}{\tiny$+[-0.048-0.136\ln(x+0.05)-1.384\ln(y+0.05)](x+0.05)^{6.282+0.078\ln(x+0.05)}(y+0.05)^{3.436+0.583\ln(y+0.05)}$}\\
\multicolumn{6}{l}{\tiny$+[0.430+0.726\ln(x+0.05)+0.144\ln(y+0.05)](x+0.05)^{6.046+0.978\ln(x+0.05)}(y+0.05)^{3.357+0.485\ln(y+0.05)}$}\\
\hline
\raisebox{-0\height}[0.03\textwidth]{Sample} &$\hat M$ & $\boldsymbol{\hat\mu}_a$ & $\boldsymbol{\hat\mu}_{n_x}$ & $\boldsymbol{\hat\mu}_{n_y}$ & $\boldsymbol{\hat\sigma}_a$\\ \hline
500 & 7 & \tiny $\begin{pmatrix}0.038 \\1.484\\0.643\\-1.311\\-0.461\\0.732\\-0.020\end{pmatrix}$ & \tiny $\begin{pmatrix}1.316\\1.152\\2.146\\2.220\\2.817\\ 1.404\\7.540\end{pmatrix}$ & \tiny $\begin{pmatrix}1.813\\1.470\\1.316\\1.302\\8.166\\6.254\\13.707\end{pmatrix}$ & \tiny $\begin{pmatrix}26.112\\0.258\\0.280  \\0.468 \\8.361\\0.010\\0.045\end{pmatrix}$ \\
\hline
$D$ & Run Time &$\boldsymbol{\hat\sigma}_{n_x}$ & $\boldsymbol{\hat\sigma}_{n_y}$ & $\boldsymbol{\hat\rho}_x$ & $\boldsymbol{\hat\rho}_y$ \\
\hline
0.078 & 183.06 s& \tiny $\begin{pmatrix}0.726\\0.033\\0.003\\0.960\\0.137\\0.235\\1.209\end{pmatrix}$& \tiny $\begin{pmatrix}0.424 \\0.279\\0.337\\0.870\\0.673\\0.210\\0.520\end{pmatrix}$& \tiny $\begin{pmatrix}0.104\\-0.154\\0.987\\0.985\\0.587 \\0.358\\-0.317\end{pmatrix}$& \tiny $\begin{pmatrix}-0.300\\-0.968\\-0.968\\-0.450\\0.107\\-0.191\\-0.169\end{pmatrix}$ \\
\hline
\multicolumn{6}{c}{Estimator} \\
\hline
\multicolumn{6}{l}{\tiny$\hat{f}(x,y)=[0.038+1.981\ln(x+0.1)-3.319\ln(y+0.1)](x+0.1)^{1.316+0.363\ln(x+0.1)}(y+0.1)^{1.813+0.212\ln(y+0.1)}$}\\
\multicolumn{6}{l}{\tiny$+[1.484-0.001\ln(x+0.1)-0.070\ln(y+0.1)](x+0.1)^{1.152+0.017\ln(x+0.1)}(y+0.1)^{1.470+0.139\ln(y+0.1)}$}\\
\multicolumn{6}{l}{\tiny$+[0.643+0.001\ln(x+0.1)-0.091\ln(y+0.1)](x+0.1)^{2.146+0.001\ln(x+0.1)}(y+0.1)^{1.316+0.168\ln(y+0.1)}$}\\
\multicolumn{6}{l}{\tiny$+[-1.311+0.442\ln(x+0.1)-0.183\ln(y+0.1)](x+0.1)^{2.220+0.480\ln(x+0.1)}(y+0.1)^{1.302+0.435\ln(y+0.1)}$}\\
\multicolumn{6}{l}{\tiny$+[-0.461+0.673\ln(x+0.1)+0.604\ln(y+0.1)](x+0.1)^{2.817+0.069\ln(x+0.1)}(y+0.1)^{8.167+0.336\ln(y+0.1)}$}\\
\multicolumn{6}{l}{\tiny$+[0.732+0.001\ln(x+0.1)-0.0004\ln(y+0.1)](x+0.1)^{1.404+0.118\ln(x+0.1)}(y+0.1)^{6.254+0.105\ln(y+0.1)}$}\\
\multicolumn{6}{l}{\tiny$+[-0.020-0.017\ln(x+0.1)-0.004\ln(y+0.1)](x+0.1)^{7.540+0.605\ln(x+0.1)}(y+0.1)^{13.707+0.26\ln(y+0.1)}$}\\
\hline
\end{tabularx}
\caption{Case 1: $f(x,y)=e^{-x^2+y}$, observed over the unit square $[0,1]^2$%
, estimated over $[0,1.2]^2$, and set $x_0=y_0=-0.05$, $\protect\sigma=0.5$
and $M_{\mbox{max}}=6$.\newline
Case 2: $f(x,y)=x^3y-y^2e^x+3xy$, over the unit square $[0,1]^2$, estimated
over $[0,1.2]^2$, and set $x_0=y_0=-0.1, \protect\sigma=0.05$, and $M_{
\mbox{max}}=8$.}
\label{t5.3}
\end{table}

For the first case, let $f(x,y)=e^{-x^2+y}$, observed over the unit square $%
[0,1]^2$, and set $x_0=y_0=-0.05$, $\sigma=0.5$ and $M_{\mbox{max}}=6$. We
will assess the performance of the MTPE over the rectangle $[0,1.2]^2$.
Results are shown in Table \ref{t5.3} (top), and in Figure \ref{f5.3.1}
(left).

\begin{figure}[tbp]
\includegraphics[width=0.45\textwidth]{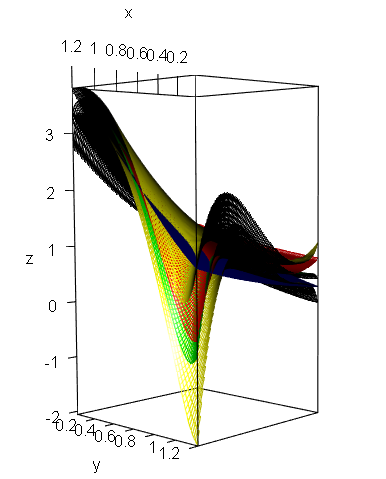} %
\includegraphics[width=0.45\textwidth]{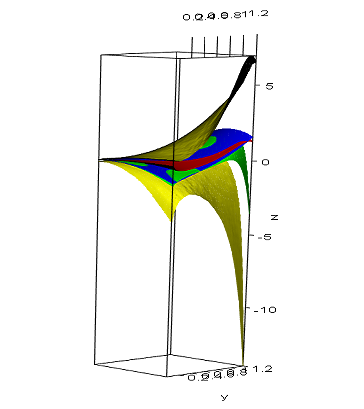}
\caption{ Left Plot: $f(x,y)=e^{-x^2+y}$. Right Plot: $%
f(x,y)=x^3y-y^2e^x+3xy $. We have $(x,y)\in [0,1.2]^2$ in both cases.
Displaying the estimators (red), the means of point process realizations
(green), and the true surfaces (blue), along with the 95\% envelope surfaces
(yellow), for a sample size $K=500$.}
\label{f5.3.1}
\end{figure}

The second case we consider is the function $f(x,y)=x^3y-y^2e^x+3xy$, over
the unit square $[0,1]^2$, and set $x_0=y_0=-0.1, \sigma=0.05$, and $M_{%
\mbox{max}}=8$. We assess the performance of the MTPE over the set $%
[0,1.2]^2 $. Results are shown in Table \ref{t5.3} (bottom), and in Figure %
\ref{f5.3.1} (right).

In both cases, the MTPE performs well, with integrated distances $0.032$ and
$0.078$, respectively, over the set $[0,1.2]^2$. Although it is hard to see
all the surfaces in both Figures, we can clearly see how the MTPE surface
(red) is almost identical to the true surface (blue). We can further see how
the envelope surfaces (yellow) contain the true surface.

We assess the performance of the TPE and MTPE against other commonly used
methods, in the following.

\textbf{Comparison to Other Methods}: In order to further assess the
performance of the proposed method, we conducted comparisons of the TPE and
MTPE against other commonly used function approximation methods, including
kernel quantile regression (KQR), Gaussian process (GP), spline regression,
and neural network (NN). In all the simulations that follow we redraw the
observed points, apply each method to recover the estimate, calculate the
integrated distances from the true function, and then present the results.

We conducted the estimation procedures using established R packages; for the
kernel quantile regression and Gaussian process methods, we used functions
\textit{kqr} and \textit{gausspr} from the R package \textit{kernlab}; for
spline regression, we used the function \textit{bs} from the R package
\textit{splines} and for the neural network approach, we used function
\textit{dnn} from the R package \textit{cito}.

\begin{figure}[tbp]
\centering \includegraphics[width=0.3\textwidth]{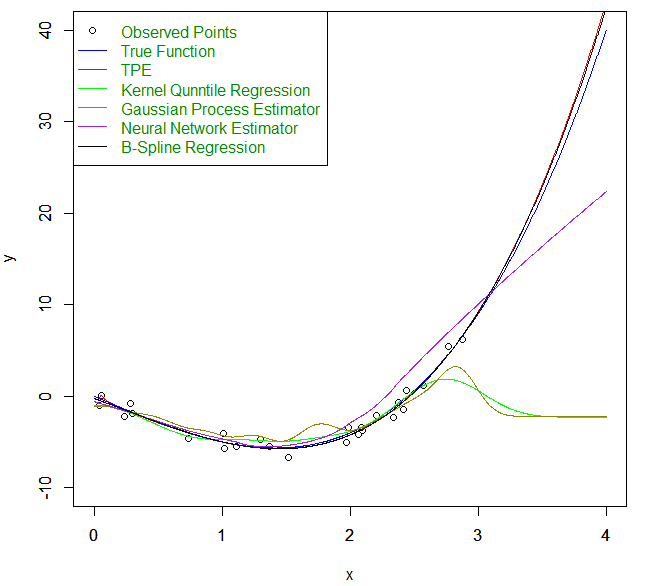} %
\includegraphics[width=0.3\textwidth]{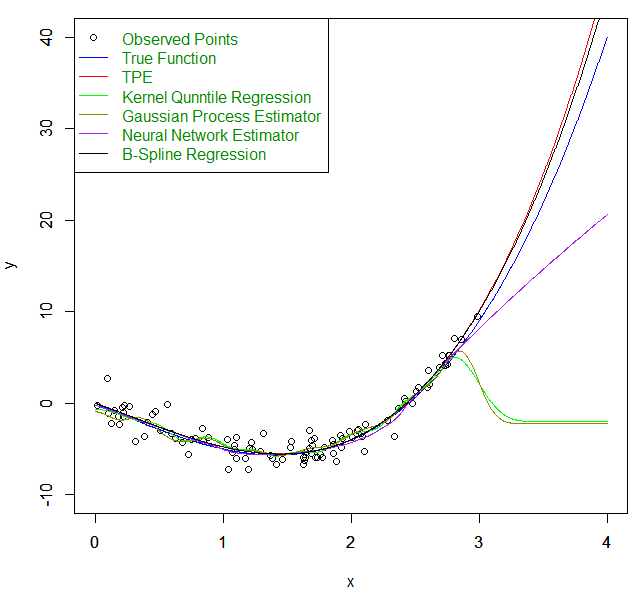} %
\includegraphics[width=0.3\textwidth]{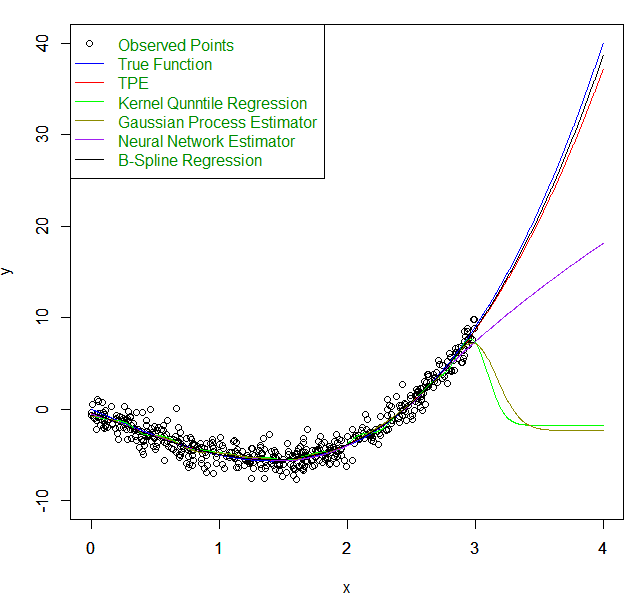}
\caption{Comparisons of the estimators from different methods versus the
true function $f(x)=x^3-6x$, $x\in (0,4]$. We used sample sizes $K=25$
(left), $K=100$ (middle), and $K=500$ (right). This is the ideal case for
spline regression, and this method performs as expected giving a near
perfect fit. The TPE is clearly performing better than any other method.}
\label{ComparisonPlot}
\end{figure}

\begin{figure}[tbp]
\centering \includegraphics[width=0.3\textwidth]{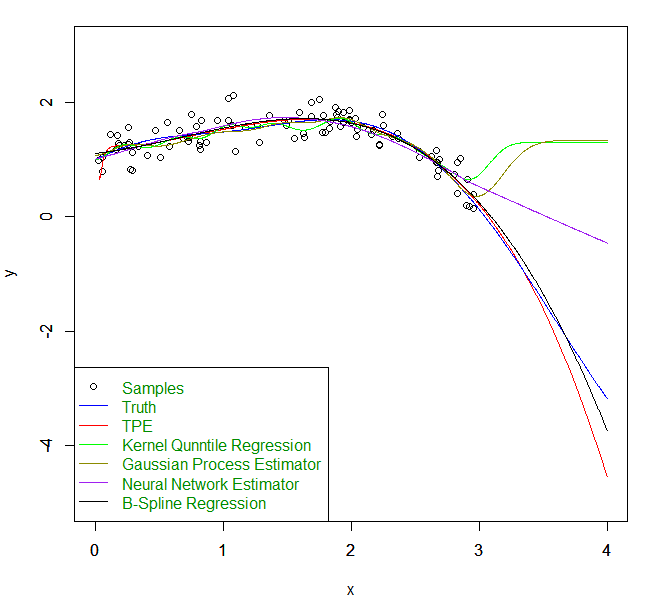} %
\includegraphics[width=0.3\textwidth]{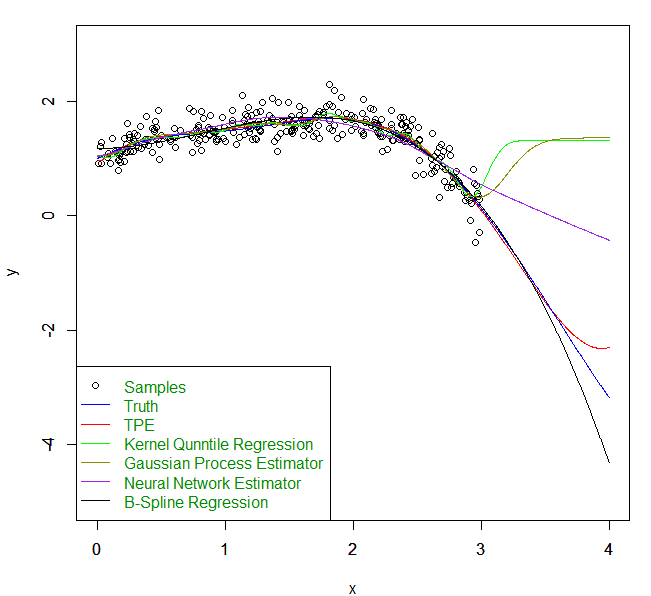} %
\includegraphics[width=0.3\textwidth]{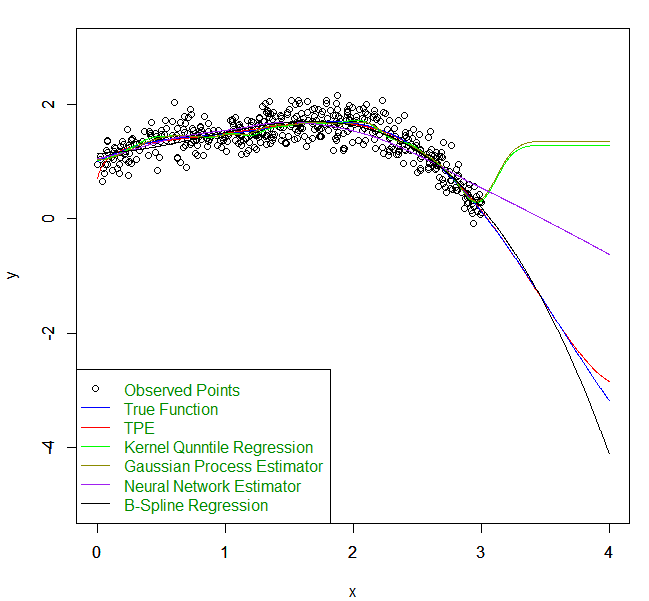}
\caption{Comparisons of the estimators from different methods versus the
true function $f(x)=x\sin(x)+e^{-x^2}+x\cos(x)/(x^2+1)$, $x\in (0,4]$. We
used sample sizes $K=100$ (left), $K=300$ (middle), and $K=500$ (right). The
TPE clearly outperforms any other method, expect for the small sample size
case, where NN is better.}
\label{ComparisonPlot2}
\end{figure}

\begin{table}[tbp]
\centering
\begin{tabularx}{\textwidth}{XXXXXX}
\hline
Methods & TPE & KQR & GP & B-Spline & NN\\
\hline
$K=25$ & 0.6363 &159.9701 &173.5798 & 0.4199 & 175.4833\\
$K=100$ & 3.7533 & 172.3583 & 175.821 & 2.1701 & 24.4164\\
$K=500$ & 0.6239 & 163.6172 & 161.7625 & 0.1908 &37.4840\\
\hline
\end{tabularx}
\caption{Integrated distances for estimators based on different methods
against the true function $f(x)=x^3-6x$, $x\in (0,4]$. This is the ideal
case for spline regression, and this method performs as expected giving a
near perfect fit, with the TPE performing better than the remaining methods.}
\label{ComparisonTable}
\end{table}

We begin with the function $f(x)=x^{3}-6x$, which is the ideal case for
spline regression (true function is a polynomial). We used sample sizes $%
K=25 $, $100$, and $500$, and the results are shown in Table \ref%
{ComparisonTable} and Figure \ref{ComparisonPlot}. We notice that all
methods work well for interpolation ($x\in (0,3]$), however, for
extrapolation ($x\in (0,4]$), there are significant differences between them.

In particular, kernel quantile regression and Gaussian process, are known to
perform poorly for extrapolation, and we observe this fact in all the
comparison examples. The neural network approach has the potential for good
prediction, but it fails to capture the truth in this first case. While the
TPE works well, under both interpolation and extrapolation scenarios, the
spline regression approach works the best in this case. This is expected
since spline regression fits polynomial models to the data, and the function
$f(x)=x^{3}-6x$, is the ideal case.

\begin{table}[tbp]
\centering
\begin{tabularx}{\textwidth}{XXXXXX}
\hline
Methods & TPE & KQR & GP & B-Spline & NN\\
\hline
$K=100$ & 0.0701 & 2.1653 & 2.0938 & 0.7108 & 0.0126\\
$K=300$ & 0.0197 & 2.1750 & 2.0936 & 0.7251 & 0.0502\\
$K=500$ & 0.0029 & 2.0922 & 2.1797 & 0.0301 &0.6472\\
\hline
\end{tabularx}
\caption{Integrated distances for estimators based on different methods
against the true function $f(x)=x\sin(x)+e^{-x^2}+x\cos(x)/(x^2+1)$, $x\in
(0,4]$. The TPE clearly outperforms any other method, expect for the small
sample size case, where NN is better.}
\label{ComparisonTable2}
\end{table}

Now we turn to a more complicated function, $f(x)=x\sin (x)+e^{-x^{2}}+x\cos
(x)/(x^{2}+1)$, $x\in (0,4]$, with sample sizes $K=100,$ $300$, and $500$.
Results are shown in Table \ref{ComparisonTable2}, and Figure \ref%
{ComparisonPlot2}. In this case, we can see that for a non-polynomial
function, the TPE outperforms the spline regression approach in every case.
Except for the small sample size case, where the NN method is slightly
better, the TPE outperforms any other method.

\begin{table}[tbp]
\centering
\begin{tabularx}{\textwidth}{XXXXXX}
\hline
Methods & TPE & KQR & GP & B-Spline & NN\\
\hline
$K=100$ & 0.0552 & 0.1666 & 0.1136 & 0.1196 & 0.0573\\
$K=300$ & 0.0230 & 0.0778 & 0.0860 & 0.1116 & 0.0420\\
$K=500$ & 0.0024 & 0.0732 & 0.0841 & 0.0343 &0.0937\\
\hline
\end{tabularx}
\caption{Integrated distances for estimators based on different methods
against the true function $f(x,y)=x^3y-y^2e^x+3xy$, $(x,y)\in (0,1.2]^2$.
The MTPE outperforms all other methods.}
\label{ComparisonTable3}
\end{table}

Furthermore, we conducted comparisons for a multivariate case, in
particular, on the function $f(x,y)=x^{3}y-y^{2}e^{x}+3xy$, $(x,y)\in
\lbrack 0,1]^{2}$, with sample size $K=100,$ $300$, and $500$. We estimate
over the set $[0,1.2]^{2}$, with results shown in Table \ref%
{ComparisonTable3}. We can clearly see that the MTPE is outperforming any
other method as the dimension increases.

\textbf{Discussion}: We ran many more examples in the univariate and
multivariate cases (omitted). We summarize our observations on obtaining the
TPE or MTPE, as well as comparisons and computational issues that arose.

First we note that as we go to higher dimensions, we immediately notice that
the time it takes to provide the MTPE is greatly increased, instead of being
just a few seconds (univariate case). This is expected since for the
univariate case, the TPE in $\Re$ with $M$ terms involves optimization over $%
5M$ parameters, whereas, the MTPE in $\Re^d$ with $M$ terms requires $%
M(3d+2) $ parameters.

Second, the envelop surfaces in the multivariate case take over a day to
calculate, instead of a few minutes or hours, as we have observed in
different univariate cases.

The comparisons performed suggest that as the sample sizes increase, all
methods provide a better fit. When it comes to extrapolation, kernel
quantile regression and Gaussian process methods are the worst, with the TPE
outperforming both the NN and spline approaches, in most cases. As dimension
increases, the MTPE outperforms all other methods.

Finally, it should be noted that in this paper our purpose was not achieving
speed of estimation, but rather accuracy of the resulting TPE or MTPE. The
code is written in R and is not optimized, where it is a well known fact
that R is extremely slow when it comes to loops, which are required to
obtain the STE and envelopes. Therefore, application of the proposed
methodology provides accurate estimators of the true function, but at the
moment it is slow as $d$ increases.

\section{Application}

In order to further exemplify the performance of the proposed methodology,
we apply our algorithm to real-life data involving international markets.

Two indices are chosen, the Dow Jones Industrial Average (DJI) index from
the United States, and the Financial Times Stock Exchange 100 Index
(FTSE100) from the United Kingdom. Both indexes are well known, and have
significant impact on global stock markets. Since they both come from
countries with similar economic systems, they are potentially highly related
to each other.

In this application, we are going to predict the FTSE100 index using both
time and the DJI index. The historical data consist of daily measurements
from January 1, 2001 to October 15, 2020 for both DJI and FTSE100 (7227
days). The values of these indices correspond to the values observed at the
close of the market for that day. We keep days when both markets were open
and this leads to a total of 4889 data points. The data were obtained from
an open source \hyperlink{https://www.investing.com/}{%
https://www.investing.com/}.

\begin{figure}[tbp]
\includegraphics[width=0.45\textwidth]{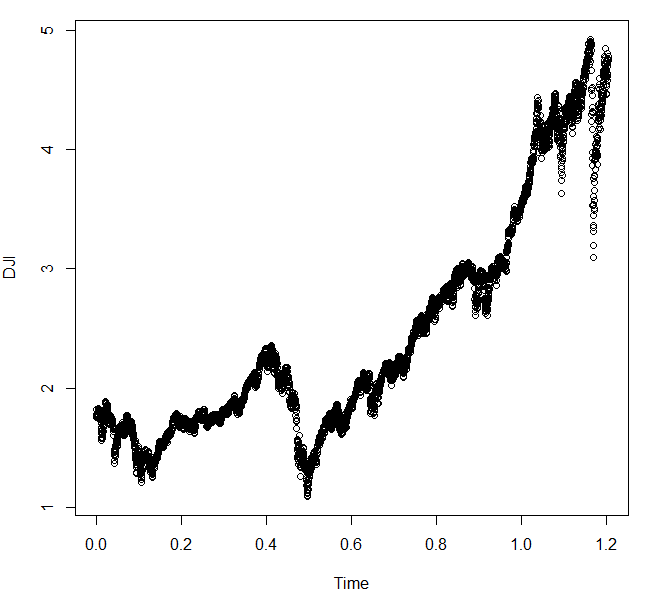} %
\includegraphics[width=0.45\textwidth]{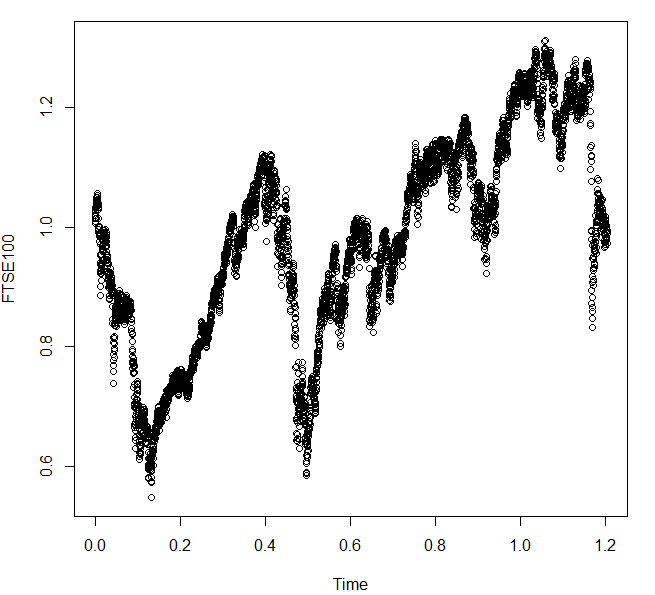}
\caption{The scatter plots of DJI (left) and FTSE100 (right) indices.}
\label{StockScatterPlot}
\end{figure}

The range of the DJI is from 6547.05 to 29551.42, and the range of the
FTSE100 is from 3287.0 to 7877.45. In order to prevent overflow problems
with the estimated parameters and the MTPE during their computation in the R
code, a re-scaling of all three variables, time, DJI, and FTSE100, is
applied, where we divided the values by 6000. After the re-scaling, the
range of the time variable (day) is from $3.33\cdot 10^{-4}$ to $1.2045$,
the range of the DJI is from $1.0912$ to $4.9252$, and the range of FTSE100
is from $0.5478$ to $1.3129$. The data is presented in Figure \ref%
{StockScatterPlot}.

Let $t$ denote time, $x$ represent the DJI index and take the response $y$
to be the FTSE100 index. Note that no other information is considered when
it comes to the statistical model, i.e., covariates such as socio-economic,
political and geo-political status and so forth. The assumption here is that
the values for the indices observed have been affected and they are the
result of such underlying, often unobserved, factors. Therefore, we will
consider modeling and forecasting of the FTSE100 based on the time stamp and
the DJI, with the aforementioned understanding in mind. Furthermore, we note
that this is not a time series model, e.g., a VAR model where the next value
is based on the current and stationarity is a required assumption to be able
to forecast. Instead, since the estimates of the parameters of the MTPE are
based on all the past data, the estimated surface provides a well informed
and reliable fit, without any stringent assumptions on the distribution of
the response or the model parameters.

As a result, we will consider the non-linear regression model
\begin{equation}
y_k(t)=\hat{f}_{FTSE100}(t,x_k(t))+\epsilon_k,
\end{equation}
where $k=1,2,\dots,K=4889$, and run Algorithm 1 to the data with $M_{max}=20$%
, $\boldsymbol{x}_0=(3.33\cdot 10^{-4}-0.1, 1.0912-0.1)=(-0.10033,0.9912)$
and $d=2$.

\begin{figure}[tbp]
\centering \includegraphics[width=0.6\textwidth]{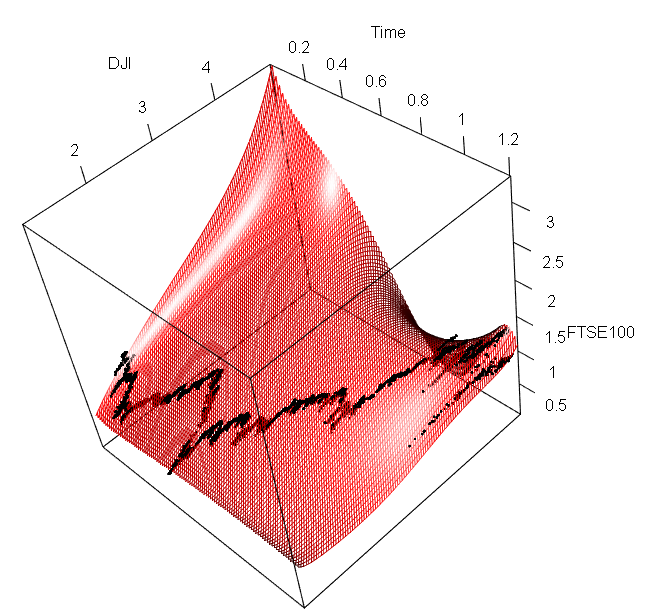}
\caption{The MTPE surface (red) along with the scatter plot of the data.}
\label{Est3D}
\end{figure}

\begin{table}[tbp]
\centering
\begin{tabularx}{\textwidth}{ccccccccc}
\hline
\raisebox{-0\height}[0.03\textwidth]{$\hat M$} & $\boldsymbol{\hat\mu}_a$ & $\boldsymbol{\hat\mu}_{n_x}$ & $\boldsymbol{\hat\mu}_{n_y}$ & $\boldsymbol{\hat\sigma}_a$ & $\boldsymbol{\hat\sigma}_{n_x}$ & $\boldsymbol{\hat\sigma}_{n_y}$ &
$\boldsymbol{\hat\rho}_x$ & $\boldsymbol{\hat\rho}_y$\\\hline
10 & \tiny$\begin{pmatrix}5.3522\cdot 10^{-1}\\2.3872\cdot 10^{-2}\\-4.4588\cdot 10^{-4}\\-4.9950\cdot 10^{-5}\\5.5904\cdot 10^{-7}\\2.7186\cdot 10^{-1}\\-5.1202\cdot 10^{-3}\\-8.1683\cdot 10^{-4}\\4.3135\cdot 10^{-5}\\-2.7525\cdot 10^{-6}\end{pmatrix}$ &\tiny $\begin{pmatrix}1.5184\\3.9855\\3.4918\\10.3921\\8.3514\\0.9958\\1.6435\\1.8004\\6.9071\\7.8058\end{pmatrix}$ & \tiny $\begin{pmatrix}0.6361\\3.4001\\6.0060\\6.7221\\8.3977\\0.5731\\3.3895\\4.7506\\6.8695\\8.6800\end{pmatrix}$ &\tiny $\begin{pmatrix}6.1832\cdot 10^2\\2.0770\\2.5377\cdot 10^{-1}\\6.2833\cdot 10^{-2}\\9.3164\cdot 10^{-3}\\2.3394\cdot 10^2\\7.8794\\1.3607\\3.9865\cdot 10^{-2}\\5.9759\cdot 10^{-3}\end{pmatrix}$ & \tiny $\begin{pmatrix}0.2628\\1.0157\\0.1917\\0.3380\\0.4868\\ 0.7076\\0.5370\\0.7344\\0.5873\\0.5096\end{pmatrix}$& \tiny $\begin{pmatrix}0.5468\\0.2480\\0.5755\\0.3165\\0.3218\\ 0.3796\\0.5225\\0.4965\\0.4909\\0.4768\end{pmatrix}$
&\tiny$\begin{pmatrix}-4.1184\cdot 10^{-3}\\-2.0508\cdot 10^{-2}\\4.4594\cdot 10^{-2}\\-3.2427\cdot 10^{4}\\2.5507\cdot 10^{-3}\\-6.9476\cdot 10^{-3}\\1.3241\cdot 10^{-2}\\-1.1238\cdot 10^{-2}\\2.2566\cdot 10^{-3}\\2.7124\cdot 10^{-3}\end{pmatrix}$& \tiny$\begin{pmatrix}-5.1690\cdot 10^{-4}\\-5.0281\cdot 10^{-3}\\-2.8144\cdot 10^{-3}\\6.2150\cdot 10^{-4}\\1.5521\cdot 10^{-4}\\-1.2874\cdot 10^{-3}\\9.6157\cdot 10^{-4}\\8.5039\cdot 10^{-4}\\5.5722\cdot 10^{-3}\\-6.5754\cdot 10^{-5}\end{pmatrix}$
\\
\hline
\multicolumn{9}{c}{Estimator} \\
\hline
\multicolumn{9}{l}{\small$\hat{f}_{FTSE100}(t,x)= \left(0.5.352-0.749\ln(t+0.1)-0.196\ln(x-0.991)\right) (t+0.1)^{1.518+0.035\ln(t+0.1)}$}\\
\multicolumn{9}{l}{\small$(x-0.991)^{0.636+0.150\ln(t+0.1)}+ \left(2.387\cdot10^{-2}-4.327\cdot10^{-2}\ln(t+0.1)- 2.590\cdot10^{-3}\ln(x-0.991)\right)$}\\
\multicolumn{9}{l}{\small$(t+0.1)^{3.986+0.516\ln(t+0.1)} (x-0.991)^{3.400+0.031\ln(t+0.1)}$}\\
\multicolumn{7}{l}{\small$+\left(-4.459\cdot10^{-4}+2.169 \cdot10^{-3}\ln(t+0.1)-4.11\cdot10^{-4}\ln(x-0.991)\right)$}\\
\multicolumn{9}{l}{\small$(t+0.1)^{3.492+0.018\ln(t+0.1)}(x-0.991)^{6.006+0.166\ln(t+0.1)}$}\\
\multicolumn{9}{l}{\small$+\left(-4.995\cdot10^{-5}-6.887\cdot10^{-6}\ln(t+0.1)-1.236\cdot10^{-5}\ln(x-0.991)\right)(t+0.1)^{10.39+0.057\ln(t+0.1)}$}\\
\multicolumn{9}{l}{\small$(x-0.991)^{6.722+0.05\ln(t+0.1)} +\left(5.89\cdot10^{-7}-1.157\cdot10^{-5}\ln(t+0.1)-4.653 \cdot10^{-7}\ln(x-0.991)\right)$}\\
\multicolumn{9}{l}{\small$(t+0.1)^{8.351+0.118\ln(t+0.1)} (x-0.991)^{8.398+0.052\ln(t+0.1)}$}\\
\multicolumn{9}{l}{\small$+\left(0.2719-1.150\ln(t+0.1)-0.114\ln(x-0.991)\right)(t+0.1)^{0.996+0.250\ln(t+0.1)}$}\\
\multicolumn{9}{l}{\small$(x-0.991)^{0.573+0.072\ln(t+0.1)}+\left(-5.120\cdot10^{-3}+5.603\cdot10^{-2}\ln(t+0.1)+3.959\cdot10^{-3}\ln(x-0.991)\right)$}\\
\multicolumn{9}{l}{\small$(t+0.1)^{1.644+0.144\ln(t+0.1)}(x-0.991)^{3.390+0.136\ln(t+0.1)}$}\\
\multicolumn{9}{l}{\small$+\left(-8.168\cdot10^{-4}-1.123\cdot10^{-2}\ln(t+0.1)+5.744\cdot10^{-4}\ln(x-0.991)\right)$}\\
\multicolumn{9}{l}{\small$(t+0.1)^{1.800+0.270\ln(t+0.1)}(x-0.991)^{4.751+0.123\ln(t+0.1)}$}\\
\multicolumn{9}{l}{\small$+\left(4.313\cdot10^{-5}+5.283\cdot10^{-5}\ln(t+0.1)+1.090\cdot10^{-4}\ln(x-0.991)\right)$}\\
\multicolumn{9}{l}{\small$(t+0.1)^{6.907+0.172\ln(t+0.1)}(x-0.991)^{6.870+0.120\ln(t+0.1)}$}\\
\multicolumn{9}{l}{\small$+\left(-2.753\cdot10^{-6}+8.260\cdot10^{-6}\ln(t+0.1)-1.873\cdot10^{-7}\ln(x-0.991)\right)$}\\
\multicolumn{9}{l}{\small$(t+0.1)^{7.806+0.130\ln(t+0.1)}(x-0.991)^{8.680+0.114\ln(t+0.1)}$}\\
\hline
\end{tabularx}
\caption{Stock Market Application: MTPE parameters for the fitted surface $%
\hat{f}_{FTSE100}(t,x_k(t))$.}
\label{ApplicationResult}
\end{table}
\newpage

The results of the estimation procedure are given in Table \ref%
{ApplicationResult}. The fitted model yields $\hat{M}=10$, with the
estimated MTPE surface of the FSTE100 given in the bottom of the table. In
Figure \ref{Est3D} we present the MTPE surface (red) along with the scatter
plot of the data. The estimated surface fits the data very well but it
cannot be visualized well in the 3d plot.

\begin{figure}[tbp]
\centering \includegraphics[width=0.45%
\textwidth]{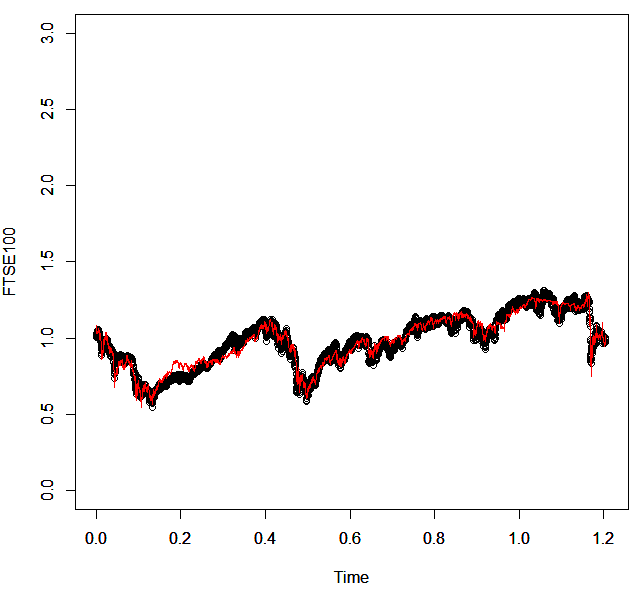} \includegraphics[width=0.45%
\textwidth]{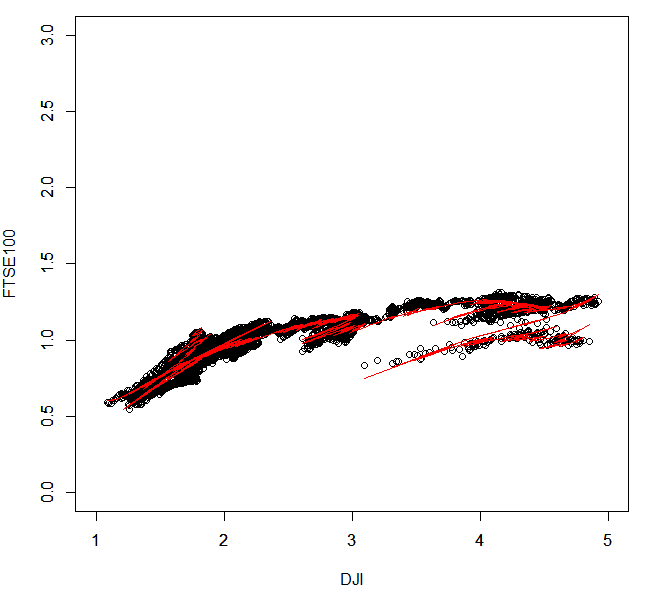}
\caption{Left Plot: We project the fitted surface onto the time axis. Red
denotes the projected fitted surface. Right plot: We project the fitted
surface onto the DJI axis. Red denotes the projected fitted surface, and
since the data is not ordered by time, we get the disjoined fits observed.}
\label{Projection}
\end{figure}

Now if we keep the value of one dimension fixed and then project the
estimated surface onto the other dimension, we obtain the plots given in
Figure \ref{Projection}. In both projections, the MTPE fits the data well
(red lines). When we project the surface onto the time axis, we could see a
good plot, since the FSTE100 is ordered (indexed) by time. In this case, for
any given time, there is only one corresponding value of the FSTE100.
However this is not the case when we project the surface onto the DJI axis.
For a given DJI, there could be more than one corresponding values of
FSTE100, and therefore there will be overlaps and disjoined lines in the
plot.

\begin{figure}[tbp]
\centering
\includegraphics[width=0.45\textwidth]
{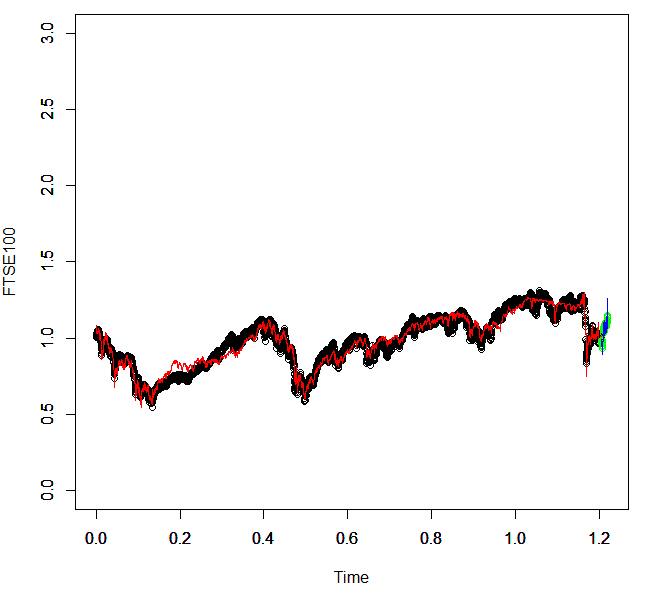} \includegraphics[width=0.45%
\textwidth] {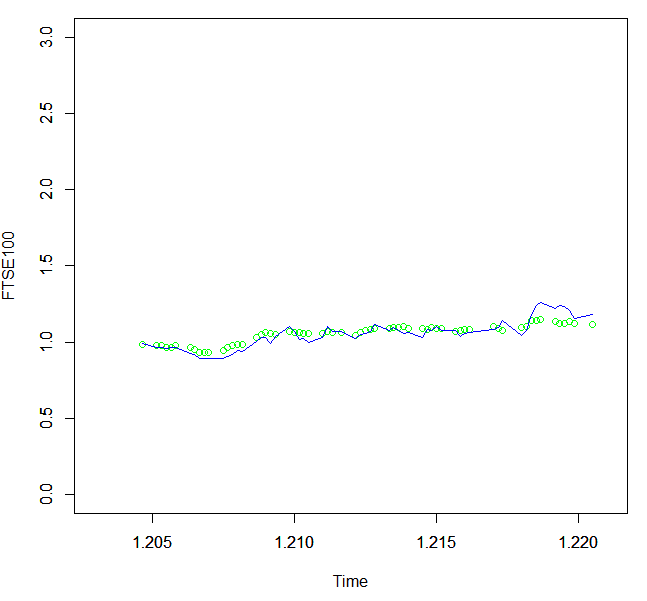}
\caption{Left Plot: Red denotes the projection of the fitted surface to the
time axis, green corresponds to the true values, and blue is the forecast
for the last 3 months. Right Plot: The green points denote the true index
values for the last 3 months, and blue denotes the projection of the surface
forecast.}
\label{ProjectionExtend}
\end{figure}

In order to further investigate the performance of our methodology, we drop
the last three months from the data and perform forecasting. More precisely,
using the original range of the data (with the last three months removed,
treating them as unobserved), we obtain the MTPE in the original range of
the data, that is, we perform forecasting for the last three months. The
results are shown in Figure \ref{ProjectionExtend} (projections on the time
axis). On the right plot, we can clearly see that the forecasted MTPE (blue)
fits the true data very well, where the true, assumed unobserved, values of
the FTSE100 (green points) for the next three months are also displayed. As
expected, the further away from the observed data, the extrapolation
performance diminishes (times from about 1.2175 to 1.220, i.e., the last
month), however, for about the first two months, the MTPE performs
exceptionally.

\section{Concluding Remarks}

We have proposed and studied a novel non-linear regression framework, where
the response depends on predictors via a general function $f$, not analytic
necessarily. We defined the Stochastic Taylor Expansion as a generalization
of Taylor's theorem, and utilized point process theory to obtain the MTPE $%
\hat{f}$ of the true function $f$. We further proved that the MTPE converges
to the true function uniformly almost surely.

Our simulations illustrated that the proposed methodology is able to recover
the true function consistently within the range of the data (interpolation).
In terms of extrapolation, as we observed in our simulation section, moving
away from the observed data the estimated function performance diminishes,
but this is a common issue when forecasting. However, if the true function
is analytic and it's Taylor expansion requires a finite number of terms, the
methodology is able to provide a near perfect fit even outside the range of
data, provided that the sample size is large. This fact indicates that the
methods proposed herein can form the fundamental framework that will allow
us to perfectly recover a function even outside the range of the observed
data. Further details on this approach will be forthcoming.

The model used for the underlying Poisson point process involved a flexible
mixture of normals intensity function. Generalizations to other point
process models, such as Gibbs and Cox point processes, can also be used in
order to introduce dependence or conditional independence, respectively,
between the coefficients and powers of the STE. From a mathematical point of
view, the MTPE can be used to create a space of functions, that was proven
to be dense in the space of continuous functions. Additional investigation
is required in order to connect the created MTPE function space with
reproducing kernel Hilbert spaces. These are subjects of great interest that
will be investigated elsewhere.

\bibliographystyle{plainnat}
\bibliography{TaylorIPPRef}

\end{document}